\documentclass[12pt]{article}
\usepackage[utf8]{inputenc}
\usepackage{authblk}
\usepackage[T1]{fontenc}
\usepackage[a4paper,left=2cm,right=2cm,top=3cm,bottom=3cm]{geometry}
\usepackage[sorting=none]{biblatex}
\usepackage{bbold}
\usepackage[english]{babel}
\usepackage{xcolor}
\usepackage{libertine}
\usepackage{graphicx}
\usepackage{amsmath,amsfonts,amssymb,bm}
\usepackage{array}
\usepackage{enumerate}
\usepackage{amsfonts,verbatim,mathdots}
\usepackage{afterpage}
\usepackage{thmtools, thm-restate}
\usepackage{amsthm}
\usepackage{mathrsfs}
\usepackage{stmaryrd}
\usepackage[export]{adjustbox}
\usepackage{caption}
\usepackage{subcaption}
\usepackage{slashed}
\usepackage{enumitem}
\usepackage[colorlinks,linkcolor=blue,citecolor=red]{hyperref}
\usepackage[toc]{appendix}
\usepackage{mathtools}
\usepackage{etoolbox}

\graphicspath{{figures/}}

\def\R{{\mathbb R}}\def\C{{\mathbb C}}\def\Z{{\mathbb Z}}\def\N{{\mathbb N}}

\usepackage{fancyhdr}

\newtheorem{lemma}{Lemma}[section]
\newtheorem{proposition}{Proposition}[section]
\newtheorem{theorem}{Theorem}[section]
\newtheorem*{remark}{Remark}
\newtheorem*{remarks}{Remarks}
\newtheorem{definition}{Definition}[section]

\makeatletter
\patchcmd{\maketitle}{\@fnsymbol}{\@alph}{}{}  % Footnote numbers from symbols to small letters
\makeatother

\newcommand{\monthyeardate}{\ifcase \month \or January\or February\or March\or %
April\or May \or June\or July\or August\or September\or October\or November\or %
December\fi, \number \year}

\font\myfont=cmr12 at 40pt

\title{{\myfont Triviality of mean-field 
$\varphi^4_4$-theories using the flow equations}}

\title{Triviality proof for mean-field 
$\varphi^4_4$-theories}

\author{
  Pierre WANG\thanks{pierre.wang@polytechnique.edu}
  , Christoph KOPPER\thanks{christoph.kopper@polytechnique.edu}
}

\affil{Centre de Physique Théorique (CPHT), CNRS, UMR 7644
Institut Polytechnique de Paris, 91128 Palaiseau, France}

\begin{document}

\date{\monthyeardate}

\maketitle

\begin{abstract}
The differential equations of the Wilson renormalization group are a
powerful tool to study the Schwinger functions of Euclidean quantum 
field theory. In particular renormalization theory can 
be  based entirely on inductively bounding their perturbatively expanded  
solutions. Recently the solutions of these equations for scalar field theory
have been analysed rigorously without recourse to perturbation theory, 
at the cost of restricting to the mean-field approximation \cite{Kopper2022}. 
In particular it was shown there that one-component $\varphi^4_4$-theory is trivial 
if the bare coupling constant of the UV regularized theory is not large. 
This paper presents progress w.r.t. \cite{Kopper2022}:\\ 
1. The upper bound on the bare coupling is sent to infinity and the proof
is extended to  $O(N)$ vector models.\\
2. The unphysical infrared cutoff used in \cite{Kopper2022} 
for technical simplicity is replaced by a physical mass.   
\end{abstract}

\section{Introduction}

\paragraph{}Quantum field theory is the fundamental framework of theoretical
 physics. It comprises both quantum mechanics and special relativity and acts 
as a powerful tool to study systems with a large or infinite number of degrees 
of freedom. Euclidean field theory is used in statistical mechanics in order 
to study critical behavior. Relativistic field theory is related to  
the Euclidean theory via analytic continuation. In perturbative quantum field 
theory, one typically expands the correlation functions which appear in 
transition 
amplitudes, in a power series w.r.t. a coupling constant $\lambda$ which 
represents 
the strength of the interaction. Feynman graph amplitudes are  contributions 
to these correlation functions, they are typically of order $\lambda^V$ if the number 
of vertices is V. 

Such expansions may lead to contributions which are ill-defined. E.g.
 in the $\varphi^4_4$-theory the first order correction to the two-point 
function is UV-divergent. A standard procedure is then to first 
regularize the theory, and to renormalize it afterwards through the introduction
of counterterms. In $\varphi^4_4$-theory in four dimensions, 
one adds counterterms in the lagrangian for the mass, the coupling constant 
and the wavefunction so that the previously divergent graphs eventually become 
finite. 

The differential flow equations of the Wilson renormalization group \cite{Wilson1}-\cite{Wilson2}
were introduced for the first time by Wegner and Houghton in 1972 
\cite{Wegner1972}. Polchinski in his seminal paper \cite{Polchinski1984} 
proved the perturbative 
renormalizability of the scalar $\varphi^4_4$-theory using these 
flow equations. Instead of dealing with the combinatorial complexity of 
Feynman diagrams, he analyzed the Schwinger functions as a whole. 
They are regularized by an UV-cutoff and an infrared cutoff, also called flow parameter. 
The flow equations are differential equations whose solutions are the regularized connected amputated Schwinger functions.
The solutions of the perturbative flow equations can be bounded using  
an inductive scheme. These bounds are sufficiently strong  to prove
renormalizability \cite{Keller_Kopper_Salmhofer_1990}. 
Later on, the proof of perturbative renormalizability
was extended to the massless $\varphi^4_4$-theory \cite{Keller_massless_phi^4}, 
to the non-linear $\sigma$-model \cite{Mitter_Ramadas_sigma}. The flow equations were also used to prove rigourously renormalizability of spontaneously broken $SU(2)$ Yang-Mills 
theory \cite{KOPPER_2009_SU}-\cite{Efremov_2017}, perturbative renormalizability in 
Minkowski space \cite{keller1996minkowski}. Other results in mathematical physics which have been established using the flow equations include
the convergence of the operator product 
expansion in perturbation theory \cite{Hollands_2012_OPE} and local Borel summability of the perturbation expansion for
Euclidean massive $\varphi^4_4$-theory 
\cite{Keller_local_Borel},{\cite{Kopper_Borel}}.

Our paper is concerned with the  so-called triviality of the 
 $\varphi^4$ theory in four dimensions. In the standard model, 
this theory appears as the pure Higgs sector after spontaneous 
symmetry breaking when ignoring the coupling of the Higgs field 
to the gauge fields and to the fermionic fields.
  Aizenman and Fröhlich in \cite{Aizenman1981},\cite{Aizenman2021} and \cite{Fröhlich1982} 
proved the triviality of the continuum limit of lattice regularized 
Euclidean one-component $\varphi^4$-theory %\cite{Aizenman2021}
 in $d>4$ dimensions, in the sense that the truncated four-point function
of the theory vanishes in this limit.
Recently Aizenman and Duminil-Copin extended the proof
to $d=4$ using multi-scale analysis in \cite{Copin2021}. 
The question whether the Standard model is trivial or not remains open,
in particular because the aforementioned proofs did not consider 
scalar fields coupled to other fields such as gauge fields or fermionic fields. It is worth to note that it has been shown that the continuum limit of lattice QED in dimensions greater than four is trivial \cite{QED_lattice_luscher} while the question in four dimensions remains open.

In this paper we take up and extend the work from 
\cite{Kopper2022}. We again restrict to the mean field approximation.
The paper is organized as follows. 
In Sect.\ref{Flow_equations_O(N)_model_general} 
we recall the flow equations for $O(N)$ vector models. 
The mean-field approximation 
of these flow equations is presented in Sect.\ref{MFA_flow_equations_O(N)}. 
In Sect.\ref{Analyticity_CAS_functions} we comment on the analyticity 
of the solutions of the flow equations w.r.t. the flow parameter called 
$\alpha$ and its consequences on the uniqueness of the solutions of the 
flow equations for fixed boundary conditions. In 
Sect.\ref{Triviality} we prove the triviality 
of the mean-field  $O(N)$ vector models, 
which include in particular pure $\varphi_4^4$-theory, 
for any value of the bare coupling 
using the technical IR cutoff from \cite{Kopper2022}. In Sect.\ref{Existence_trivial_solution}, we prove the existence of a trivial solution of the mean-field flow equations then in Sect.\ref{Uniqueness_trivial_solution} we prove the uniqueness of the trivial solution for fixed mean-field boundary conditions.
We end this section commenting on the large $N$ limit  
in Sect.\ref{Large_N_limit}. 
In Sect.\ref{Massive_theory} we replace the technical IR cutoff
from \cite{Kopper2022} by the physical one of a massive theory. 
In Sect.\ref{Massive_flow_equations_general} we derive the flow equations 
in this case, 
in Sect.\ref{Triviality_massive_theory} we again prove triviality of $\varphi^4_4$-theory in the new setting.

\section{Flow equations in the mean-field approximation}\label{Flow_equations_MeanFieldApproximation}

\subsection{The flow equations for the $O(N)$-model}
\label{Flow_equations_O(N)_model_general}

\paragraph{}

We want to analyse self-interacting $N$-component  vector scalar field theories
on four-dimensional Euclidean space. An  $N$-component  vector model with $O(N)$ symmetry  was first introduced by Stanley\cite{Stanley1968} to generalize the Ising model ($N=1$), the XY model ($N=2$) and the 
Heisenberg model $(N=3)$. In the continuum description the scalar field 
$\varphi$ has now $N$ real components 
$\varphi(x)=\left[\varphi_1(x),\cdots,\varphi_N(x)\right]^T\,$, 
and the lagrangian has a global $O(N)$ symmetry. Then the theory has $\Z_2$ symmetry under $\varphi\mapsto-\varphi$. The behaviour of 
the solutions of the $O(N)$ model in the large $N$ limit has been 
studied in detail in \cite{Amit}  and by Moshe and Zinn-Justin in 
\cite{Moshe2003}, carrying out an expansion in powers of $\frac{1}{N}\,$. 
The result is a non-trivial theory which turns out to be the exact solution 
of the spherical model \cite{Stanley_spherical_model_limit}. 
We will follow the steps from 
\cite{Kopper2022} to derive the flow equations for the $O(N)$-model
 and perform the mean-field approximation afterwards. 
Using $O(N)$ symmetry, the form of the mean-field flow equations 
generalizes those of the single-component theory.

To derive the flow equations, we base ourselves on
 \cite{Mueller2003}, \cite{Kopper2022}. 
We adopt the following convention and 
shorthand notation for the Fourier transform 
\[
f(x)=\int_p e^{ipx}\hat{f}(p)\ ,\quad \int_p:=\int\frac{d^4p}{(2\pi)^4}\ .
\]
This implies for the functional derivatives 
$\frac{\delta}{\delta\varphi(x)}$ 
\[
\dfrac{\delta}{\delta\varphi(x)}= (2\pi)^4\int_p e^{-ipx}
\dfrac{\delta}{\delta\hat{\varphi}(p)}\  .
\]

We introduce the following regularized propagator
\begin{equation}\label{Propagator_reg_1}
\hat{C}^{\alpha_0,\alpha}(p,m):=
\dfrac{1}{p^2+m^2}\Big(\exp (-\alpha_0(p^2+m^2))-\exp (-\alpha(p^2+m^2))\Big)
\geq 0\;,
\end{equation}
where $m$ is the mass of the field. Here 
$\alpha_0>0$ acts as an ultraviolet cutoff 
and $\alpha\in[\alpha_0,+\infty)$ is the flow parameter. 
By taking the limits 
$\alpha_0\rightarrow 0$ and $\alpha\rightarrow +\infty$ 
we recover the usual Euclidean propagator in momentum space, namely 
\begin{equation}
  \lim\limits_{\alpha\rightarrow +\infty}\,\lim\limits_{\alpha_0\rightarrow 0}  
\hat{C}^{\alpha_0,\alpha}(p,m)=\frac{1}{p^2+m^2}\ .
\end{equation}

With the chosen convention of the Fourier transform, the regularized propagator 
in position space $C^{\alpha_0,\alpha}(x-y,m)$ (also called covariance) reads
\begin{equation}
   C^{\alpha_0,\alpha}(x-y,m)=\int_p e^{ip(x-y)} \hat{C}^{\alpha_0,\alpha}(p,m)\ .
\end{equation}

The regularized  propagator of the $N$-component 
massive vector scalar field theory 
in momentum space, is then given by  a diagonal  $N\times N$ matrix to be called $C^{\alpha,\alpha_0}_N (x-y,m)$. 
Its elements read
\begin{equation}\label{propagator_0(N)}
    {\hat C}_{N;\;ij}^{\alpha_0,\alpha}(p,m)={\hat C}^{\alpha_0,\alpha}(p,m)\delta_{ij}\;.
\end{equation}
 This regularized propagator is positive and satisfies
\begin{itemize}
    \item ${\hat C}_{N;\;ij}^{\alpha_0,\alpha}(p,m)$ is analytic  w.r.t. $\alpha$.
   \item ${\hat C}_{N;\;ij}^{\alpha_0,\alpha_0}(p,m)=0$ \footnote{The corresponding Gaussian
measure corresponds in this case to a $\delta$-type measure on function space.}.
  \item At $\alpha$ and $i$ fixed,  ${\hat C}_{N;\;ii}^{\alpha_0,\alpha}(p,m)$ falls off
 more rapidly than any power of $\vert p\vert$.
\end{itemize}
We will consider bare interactions of the form
\begin{equation}
    L^\mathcal{V}_{0,N}(\varphi)=\int_{\mathcal{V}} d^4x\Big[b_0(\alpha_0)
\sum_{1\leq i\leq N}(\partial\varphi_i(x))^2
+\sum_{n\in 2\N}c_{0,n}(\alpha_0)\varphi(x)^n\Big]\;,
\end{equation}
where $\varphi^{2n}(x):=(\varphi^2(x))^n$ for $n\geq 1$,
$\varphi^{2}(x):=\sum_{i=1}^N\varphi_i^2(x)\,$, $(\partial\varphi_i(x))^2=\sum_{\mu=0}^3 (\partial_\mu\varphi_i(x))^2$ 
and $\mathcal{V}$ is a finite volume in $\R^4$. 
This bare interaction lagrangian is $O(N)$ invariant.
 The constants 
$b_0(\alpha_0)$, $c_{0,n}(\alpha_0)$ are called the bare couplings. 
The quantities in the sum for $n\geq 6$ are called  irrelevant terms,
 while $b_0(\alpha_0),c_{0,2}(\alpha_0)$ and $c_{0,4}(\alpha_0)$ are 
relevant terms. Generally the relevant terms are required 
in order that the renormalized physical quantities such as the 
renormalized mass or the renormalized  coupling constant are finite upon 
removing the UV cutoff. In the mean-field approximation
to be considered soon the constant
$b_0(\alpha_0)$ vanishes, because in this case
 the field variable $\varphi$ becomes a constant.

The functional integral with the bare lagrangian  $ L^\mathcal{V}_{0,N}(\varphi)$
is well-defined if for some constant 
$C^\mathcal{V}_N\in\R$, depending on $\mathcal{V}$ and $N$
\begin{equation}\label{potential_bounded_from_below}
    -\infty<C^\mathcal{V}<L_{0,N}^\mathcal{V}(\varphi)\;,
\quad \varphi\in\mbox{supp}(\mu^{\alpha_0,\alpha}_N)\;,
\end{equation}
where $\mu_N^{\alpha_0,\alpha}$ is the normalized Gaussian measure associated to 
the propagator $C^{\alpha_0,\alpha}_N$. Some  
properties of Gaussian measures taylored for our purposes,
can be found in \cite{Mueller2003}, more information
can be found in \cite{Glimm1987}. We collected a few  
items in Appendix \ref{appendix_A}. 
The field $\varphi$ is supposed to belong to the support of the 
Gaussian measure $\mu^{\alpha_0,\alpha}_N$. Since the regularized propagator 
$\hat{C}^{\alpha_0,\alpha}(p,m)$ falls off more rapidly than any power of 
$\vert p\vert$ in momentum space, 
its support is contained on  smooth functions in position space, 
see e.g. \cite{Reed1973}, so that the quantities in 
$L_0^\mathcal{V}$ i.e. $\varphi^2(x),\varphi^4(x),\cdots$ are well-defined.
Here we will not discuss the infinite volume limit explicitly, 
for more details see \cite{Mueller2003}.  
It can be taken once we have passed to the 
connected amputated Schwinger functions (see below).  
We will thus drop the subscript $\mathcal{V}$.

The  correlation or Schwinger functions  are defined as
\begin{equation}
    \langle\varphi_{i_1}(x_1)\varphi_{i_2}(x_2)\cdots\varphi_{i_n}(x_n)
\rangle^{\alpha_0,\alpha}:=\dfrac{1}{Z^{\alpha_0,\alpha}_{N}}
\int d\mu^{\alpha_0,\alpha}_N(\varphi) 
e^{-L_{0,N}(\varphi)}\varphi_{i_1}(x_1)\varphi_{i_2}(x_2)\cdots\varphi_{i_n}(x_n)\;,
\end{equation}
where $d\mu^{\alpha_0,\alpha}_N$ is the Gaussian measure associated with 
the regularized propagator (\ref{propagator_0(N)}). 
The generating functional of the regularized
connected amputated Schwinger functions
(CAS) $e^{-L^{\alpha_0,\alpha}_{N}(\varphi)}$ at scale $\alpha$ satisfies 
\begin{equation}\label{effpot}
    e^{-L^{\alpha_0,\alpha}_{N}(\varphi)}
=\dfrac{1}{Z^{\alpha_0,\alpha}_{N}}\int d\mu^{\alpha_0,\alpha}_N(\phi)\
e^{-L_{0,N}(\phi+\varphi)}\ .
\end{equation}
Expanding $L^{\alpha_0,\alpha}_{N}(\varphi)$ in a formal power series 
in $\hat{\varphi}_i$ gives
\begin{equation}\label{expansion_O(N)}
    L^{\alpha_0,\alpha}_{N}(\varphi)=\sum_{n\in 2\N}\sum_{1\leq i_1,\cdots,i_n\leq N}
\int_{p_1,p_2,\cdots, p_n}\Bar{\mathcal{L}}^{\alpha_0,\alpha}_{n;i_1 i_2\cdots i_n}
(p_1,\cdots,p_n)\hat{\varphi}_{i_1}(p_1)\cdots\hat{\varphi}_{i_n}(p_n)\ .
\end{equation}
The CAS distributions $\Bar{\mathcal{L}}^{\alpha_0,\alpha}_{n;i_1 i_2\cdots i_n}
(p_1,\cdots,p_n)$ or moments of $L^{\alpha_0,\alpha}_{N}$ can be factorized 
 due to translation invariance as
\begin{equation}\label{factorization_O(N)}
\bar{\mathcal{L}}^{\alpha_0,\alpha}_{n,i_1 i_2\cdots i_n}(p_1,\cdots,p_n)
=\delta^4\Big(\sum_{i=1}^n p_i\Big)\mathcal{L}^{\alpha_0,\alpha}_{n;i_1 i_2\cdots i_n}
(p_1,\cdots,p_n),\quad p_n=-p_1-\cdots-p_{n-1}\ .
\end{equation}

The CAS functions to be called
$\,\mathcal{L}^{\alpha_0,\alpha}_{n;i_1 i_2\cdots i_n}(p_1,\cdots,p_n)$ are obtained through functional derivation
\begin{equation}
    \dfrac{(2\pi)^{4n}}{n!}\dfrac{\delta}{\delta\hat{\varphi}_{i_1}
(p_1)}\cdots\dfrac{\delta}{\delta\hat{\varphi}_{i_n}(p_n)}
L^{\alpha_0,\alpha}_{N}(\varphi)\vert_{\varphi\equiv 0}
=\delta^4\Big(\sum_{i=1}^n p_i\Big)\mathcal{L}^{\alpha_0,\alpha}_{n;i_1 i_2\cdots i_n}
(p_1,\cdots,p_n)\ .
\end{equation}
The flow equations are obtained on deriving (\ref{effpot})
w.r.t. $\alpha\,$
using the change 
of covariance formula (\ref{differentiation}) 
\begin{equation}\label{FE1_O(N)}
    \partial_\alpha L^{\alpha_0,\alpha}_{N}=\dfrac{1}{2}
\sum_{i=1}^N\Big\langle\dfrac{\delta}{\delta\varphi_{i}},
\dot{C}^{\alpha}\dfrac{\delta}{\delta\varphi_{i}}
\Big\rangle L^{\alpha_0,\alpha}_{N}-\dfrac{1}{2}\sum_{i=1}^N
\Big\langle\dfrac{\delta}{\delta\varphi_{i}}L^{\alpha_0,\alpha}_{N},
\dot{C}^{\alpha}\dfrac{\delta}{\delta\varphi_{i}}L^{\alpha_0,\alpha}_{N}
\Big\rangle\; , 
\end{equation}
where $\Dot{C}^\alpha=\partial_\alpha C^{\alpha_0,\alpha}$. 
Using (\ref{expansion_O(N)}),(\ref{factorization_O(N)}) 
in (\ref{FE1_O(N)}), the flow equations for the moments
$\mathcal{L}^{\alpha_0,\alpha}_{n;i_1 i_2\cdots i_n}$ can be written as
\begin{equation}
\begin{split}
    \partial_{\alpha}\mathcal{L}^{\alpha_0,\alpha}_{n; i_1i_2\cdots i_n}
(p_1,\cdots,p_n) &= \binom{n+2}{2}
\sum_{j=1}^N\int_k \dot{C}^\alpha(k,m)
\mathcal{L}^{\alpha_0,\alpha}_{n+2; i_1 i_2\cdots i_n jj}
(p_1,\cdots,p_n,k,-k) \\
    &-\dfrac{1}{2}\sum_{n_1+n_2=n+2}
\sum_{j=1}^N n_1 n_2\mathbb{S}
\Big[\mathcal{L}^{\alpha_0,\alpha}_{n_1; i_1i_2\cdots i_{n_1-1}j}
(p_1,\cdots,p_{n_1-1},q)\dot{C}^{\alpha}(q,m) \\
    &\mathcal{L}^{\alpha_0,\alpha}_{n_2; ji_{n_1}i_{n_1+1}\cdots i_{n}}
(-q,p_{n_1},\cdots,p_n)\Big]\;,
\end{split}
\end{equation}
with $q:=-p_1-p_2-\cdots-p_{n_1-1}$. Here $\mathbb{S}$ is a
symmetrisation  operator which permutes the pairs 
$(i_j,p_j)$. It averages over 
the permutations $\sigma \in S_n\,$ 
such that $\sigma(1)<\sigma(2)<\cdots<\sigma(n_1-1)$ 
and $\sigma(n_1)<\sigma(n_1+1)<\cdots<\sigma(n)\,$. 
Since we considered a theory with a $\mathbb{Z}_2$-symmetry, 
only even moments (in $n,n_1$ and $n_2$) are nonvanishing.

The FEs are an infinite system of non-linear differential equations, 
the solutions of which are the CAS functions. 
On imposing renormalization conditions, 
one can prove the perturbative renormalizability of the regularized 
theory through an inductive scheme, see \cite{Mueller2003} and
references given there.

\subsection{The mean field approximation for the $O(N)$-model}
\label{MFA_flow_equations_O(N)}
\paragraph{}
In the  mean field approximation, 
the $n$-point functions are assumed to be momentum independent. 
We set
\begin{equation}\label{Ais}
A^{\alpha_0,\alpha}_{n; i_1i_2\cdots i_n}:=
\mathcal{L}^{\alpha_0,\alpha}_{n; i_1 i_2\cdots i_n}(0,\cdots,0)\ . 
\end{equation}
In statistical physics the critical behaviour of Ising type systems 
 in $d>4$ dimensions is exactly obtained \cite{Aizenman2021},
\cite{Fröhlich1982}  in the mean-field approximation. 
We now derive the  mean-field flow equations
following \cite{Kopper2022}.
The mean field effective action $L_{mf}^{\alpha_0,\alpha}(\phi)$ 
is expanded as  a formal power series 
in the constant field $\varphi\in\R$ 
\begin{equation}\label{Mean_effective_action}
    L_{mf}^{\alpha_0,\alpha}(\varphi)=\sum_{n\in 2\N}A_n^{\alpha_0,\alpha}\varphi^n\ .
\end{equation}
We first make the technical simplification from \cite{Kopper2022}
and set $m=0$ in $\hat{C}^{\alpha_0,\alpha}(p,m)$. 
We then analyse the theory in the 
interval $\alpha\in [\alpha_0,\frac{1}{m^2}]$ so that the upper limit
on $\alpha$ takes the role of the infrared cutoff, 
thus replacing the mass. 
The existence of the UV-limit means that the mean-field solutions 
have a finite limit at $\alpha=1$\footnote{we choose units such that $m^2=1$.}, when the UV-cutoff 
$\frac{1}{\alpha_0}$ is sent to infinity. 

The regularized propagator then reads
\begin{equation}
    \dfrac{e^{-\alpha_0p^2}-e^{-\frac{1}{m^2}p^2}}{p^2}
\underset{p^2\ll 1}{=}\frac{1}{m^2}-\alpha_0+O(p^2)\ .
\end{equation}
Then the mean-field flow equations read
\begin{equation}\label{FE_O(N)_2}
\begin{split}
    \partial_{\alpha}A^{\alpha_0,\alpha}_{n; i_1 i_2\cdots i_n} 
&= \binom{n+2}{2}\frac{c}{\alpha^2}
\sum_{j=1}^N A^{\alpha_0,\alpha}_{n+2; i_1 i_2\cdots i_n jj} \\
&-
\dfrac{1}{2}\sum_{n_1+n_2=n+2} n_1 n_2\sum_{j=1}^N
\mathbb{S}\Big[A^{\alpha_0,\alpha}_{n_1; i_1i_2\cdots i_{n_1-1}j}
A^{\alpha_0,\alpha}_{n_2; ji_{n_1} i_{n_1+1}\cdots i_{n}}\Big],\quad c:=\frac{1}{16\pi^2} \ .
\end{split}
\end{equation}  

Without any further input, 
the flow equations (\ref{FE_O(N)_2})  
do not allow to construct 
inductively the CAS functions because we can only compute the 
contraction of $A^{\alpha_0,\alpha}_{n+2,i_1 i_2\cdots i_{n+2}}$ 
w.r.t. its last two indices from the CAS functions 
$A^{\alpha_0,\alpha}_{n'; i_1 i_2\cdots i_{n'}}$, $n'\leq n$. 
The mean-field solutions for the $\,O(N)\,$ model of the flow 
equations satisfy by assumption the following properties:
\begin{itemize}
    \item \textbf{(P1)}: $A^{\alpha_0,\alpha}_{n; i_1i_2\cdots i_n}=0$ if $n$ is odd.
    \item \textbf{(P2)}: $A^{\alpha_0,\alpha}_{n; i_1i_2\cdots i_n}$ is symmetric under 
any permutation of its indices $i_1,i_2,\cdots , i_n$ .
    \item \textbf{(P3)}: $A^{\alpha_0,\alpha}_{n; i_1 i_2\cdots i_n}$ is
$O(N)$-invariant in the following sense: 
let $O$ be an orthogonal matrix i.e. $O^TO=O^TO=I$, then 
    \begin{equation}
        O_{i_1 j_1}O_{i_2 j_2}\cdots O_{i_n j_n}
A^{\alpha_0,\alpha}_{n; j_1 j_2\cdots j_n}=A^{\alpha_0,\alpha}_{n; i_1 i_2\cdots i_n}\ .
    \end{equation}
\end{itemize}
Properties \textbf{(P1)}, \textbf{(P2)} and \textbf{(P3)} 
require knowledge on the $O(N)$-invariant symmetric tensors. Some facts are collected in Appendix \ref{Cartesian_Tensors}. We recall the symmetric part of a rank $n$-tensor \textbf{T} by
\begin{equation}\label{symmetric_part_tensor}
 T_{(i_1i_2\cdots i_n)}:=\frac{1}{n!}\sum_{\sigma\in S_n}
T_{i_{\sigma(1)}i_{\sigma(2)}i_{\sigma(n-1)}i_{\sigma(n)}}\ .   
\end{equation} 
and we define
\begin{equation}
    F_{i_1i_2\cdots i_n}:=\delta_{(i_1i_2}\delta_{i_3i_4}\cdots\delta_{i_{n-1}i_n)}
=\frac{1}{n!}\sum_{\sigma\in S_n}\delta_{i_{\sigma(1)}i_{\sigma(2)}}
\cdots \delta_{i_{\sigma(n-1)}i_{\sigma(n)}}\;.
\end{equation}

From Proposition \ref{tensor_structure}, we set
\begin{equation}
    A^{\alpha_0,\alpha}_{n; i_1 i_2\cdots i_n}=A^{\alpha_0,\alpha}_n F_{i_1i_2\cdots i_n}
\end{equation}
with $A^{\alpha_0,\alpha}_n$ smooth.

The mean field flow equations 
(\ref{FE_O(N)_2}) can now be rewritten as
\begin{equation}\label{FE_O(N)_3}
\begin{split}
    \partial_{\alpha}A^{\alpha_0,\alpha}_{n}  F_{i_1i_2\cdots i_n} 
&= \binom{n+2}{2}\frac{c}{\alpha^2}A^{\alpha_0,\alpha}_{n+2}
\sum_{j=1}^N F_{i_1i_2\cdots i_n jj} \\
    &-\dfrac{1}{2}\sum_{n_1+n_2=n+2}n_1 n_2 A^{\alpha_0,\alpha}_{n_1}
A^{\alpha_0,\alpha}_{n_2}\sum_{j=1}^N 
\mathbb{S}\Big[F_{i_1i_2\cdots i_{n_1-1}j}F_{ji_{n_1} i_{n_1+1}\cdots i_{n}}\Big]\ .
\end{split}
\end{equation}

\paragraph{} From Proposition \ref{contraction}, 
the flow equations (\ref{FE_O(N)_3}) reduce to a much simpler form

\begin{equation}\label{FE_O(N)_4}
\begin{split}
    \partial_{\alpha}A^{\alpha_0,\alpha}_{n} 
&= \binom{n+2}{2}\dfrac{N+n}{n+1}\frac{c}{\alpha^2}A^{\alpha_0,\alpha}_{n+2}
-\dfrac{1}{2}\sum_{n_1+n_2=n+2}n_1 n_2 A^{\alpha_0,\alpha}_{n_1}A^{\alpha_0,\alpha}_{n_2}\ .
\end{split}
\end{equation}

Because of the factor $\binom{n+2}{2}$, 
iterated integration w.r.t. $\alpha$ of approximate solutions
of the infinite dynamical 
system for the functions  $\,A_n^{\alpha_0,\alpha}\,$ 
would produce at each step factors of order $n^2\,$. 
Therefore this procedure is unstable w.r.t. $n$ from 
a combinatorial point of view. As a consequence, 
we will follow the strategy from \cite{Kopper2022}: 
 \begin{itemize}
     \item Start from a smooth two-point function $A_2^{\alpha_0,\alpha}$
\footnote{We will specify the properties of  two-point function $A_2^{\alpha_0,\alpha}$
later on.} and fix boundary conditions at the bare scale $\alpha=\alpha_0\,$.
     \item Smooth solutions $A_n^{\alpha_0,\alpha}$ can then be constructed inductively 
using (\ref{FE_O(N)_4}). Their properties depend on $A_2^{\alpha_0,\alpha}\,$.
    \end{itemize} 

Adopting the change of function and variable 

\begin{equation}\label{definition_f_n}
f_{n}(\mu):=n\alpha^{2-\frac{n}{2}}c^{\frac{n}{2}-1} 
A^{\alpha_0,\alpha}_{n}\;,\quad \mu:=\ln\left(\dfrac{\alpha}{\alpha_0}\right)\;,
\end{equation}
the mean-field flow equations read 
\begin{equation}\label{MFE_O(N)}
    f_{n+2}(\mu)=\dfrac{1}{n+N}
\sum_{n_1+n_2=n+2}f_{n_1}(\mu)f_{n_2}(\mu)+\dfrac{n-4}{n(n+N)}f_n(\mu)
+\dfrac{2}{n(n+N)}\partial_\mu f_n(\mu)\;,\quad n\geq 2\;,
\end{equation}
for $\mu\in[0,\mu_{\max}]\,$ where $\mu_{\max}:=\ln\left(\frac{1}{\alpha_0}\right)$.

\paragraph{} In \cite{Kopper2022} locally analytic smooth solutions
 $f_n(\mu)$, uniformly bounded w.r.t. $\mu$ 
with bare mean-field action locally analytic w.r.t. $\varphi$,
were shown to exist. 
A subclass of these solutions are smooth solutions which  
vanish at $\mu=0$ upon removing the UV-cutoff so that they are 
asymptotically free in the ultraviolet.

\begin{remark}
\begin{itemize}
    \item 
The statements in \cite{Kopper2022} on the local analyticity w.r.t. $\mu$ of uniformly bounded smooth solutions 
$f_n(\mu)$ and the analyticity of the bare-mean field action w.r.t. 
$\varphi$ remain valid for $N>1$.
\end{itemize}
\end{remark}

\subsection{Local analyticity w.r.t. $\alpha$ of the mean-field CAS functions}
\label{Analyticity_CAS_functions}
 \paragraph{} 
We recall that the regularized propagator in momentum space 
$\hat{C}^{\alpha_0,\alpha}(p,m)$ introduced in (\ref{Propagator_reg_1}) 
is analytic w.r.t. $\alpha$. If we construct the  solutions $A_{n+2}^{\alpha_0,\alpha}$ of the FEs 
(\ref{FE_O(N)_2}) as indicated, we have the following analyticity and 
uniqueness statements  
\begin{proposition}
\label{Analyticity_of_correlation_functions}

\begin{itemize}
    \item  Let $A_n^{\alpha_0,\alpha}$ be mean-field smooth solutions 
of (\ref{FE_O(N)_2}). 
The boundary conditions are assumed to be induced by 
a two-point function $A_2^{\alpha_0,\alpha}$ and its derivatives at $\alpha=\alpha_0$ which is locally analytic w.r.t. $\alpha$. 
Then $A_n^{\alpha_0,\alpha}$ is locally 
analytic w.r.t. $\alpha\,$.

\item Let $A_2^{\alpha_0,\alpha}$ and $\Tilde{A}_2^{\alpha_0,\alpha}$ be 
locally analytic w.r.t. $\alpha$.  If $A_n^{\alpha_0,\alpha}$ and 
$\Tilde{A}_n^{\alpha_0,\alpha}$ are constructed from 
$A_2^{\alpha_0,\alpha}$ and $\Tilde{A}_2^{\alpha_0,\alpha}$ respectively,
using the flow equations (\ref{FE_O(N)_2}) 
and if
\begin{equation}
\partial_\alpha ^k \tilde{A}_2^{\alpha_0,\alpha}|_{\alpha =\alpha_0}
=\partial_\alpha ^k A_2^{\alpha_0,\alpha}|_{\alpha =\alpha_0} \ ,\quad k\geq 0\;,
\end{equation}
then we have for arbitrary $\alpha$
\begin{equation}
    \partial_\alpha^k \tilde{A}_n^{\alpha_0,\alpha}
=\partial_\alpha^k A_n^{\alpha_0,\alpha},\quad k\geq 0\,,\quad  n\geq 2\ .
\end{equation}
\end{itemize}
\end{proposition}

\begin{proof}
    The proof of the first statement proceeds  by induction in $n$. 
It obviously holds for $n=2$. From (\ref{FE_O(N)_4}), we have
    \begin{equation}
        A_{n+2}^{\alpha_0,\alpha}
=\dfrac{2}{c(n+N)(n+2)}\alpha^2\partial_\alpha A_n^{\alpha_0,\alpha}
+\dfrac{\alpha^2}{c(n+N)(n+2)}\sum_{n_1+n_2=n+2} n_1 n_2 
A_{n_1}^{\alpha_0,\alpha} A_{n_2}^{\alpha_0,\alpha}\ .
    \end{equation}
which implies the statement using the induction hypothesis.

The second statement is proven by induction in $N=n+2k$,
going up in $n$ for fixed $N\,$. It then follows directly 
from the fact that locally analytic functions are uniquely defined by their
Taylor expansions within their radius of convergence,
and from the fact that sums of products of locally analytic functions
 are again locally analytic. 
\end{proof}

Due to Proposition \ref{Analyticity_of_correlation_functions} 
 locally analytic mean-field solutions $A_n^{\alpha_0,\alpha}$ are unique
 for fixed boundary conditions at the bare scale, if we  start
 from a locally analytic two-point function $A_2^{\alpha_0,\alpha}$.

\section{Triviality in the mean field approximation in the presence of an IR cutoff}\label{Triviality}

\paragraph{} Let us recall the notion of triviality in perturbative quantum field theory. From the point of view of perturbative quantum field theory, the effective coupling constant $g(\lambda)$ is a function of the energy scale $\lambda$. Its behaviour is described by the beta function defined by

\begin{equation}
    \beta(g(\lambda)):=\lambda\dfrac{dg}{d\lambda}(\lambda)\;.
\end{equation}
Note that in practice $\beta(g(\lambda))$ can only be calculated to a finite order in the perturbative expansion. In asymptotically free theories, $\beta$ is negative so that the coupling constant vanishes at high energies, $g(\lambda)\longrightarrow 0$ for $\lambda\longrightarrow +\infty$. For non-asymptotically free QFTs, such as QED or $\varphi^4_4$-theory, $\beta$ is positive. Thus the effective coupling constant grows logarithmically with $\lambda$. We define $g(0)$ the renormalized coupling constant and $g(\Lambda)$ the bare coupling where $\Lambda$ is the UV-cutoff. Keeping the value of $g(\Lambda)$ fixed and removing the UV-cutoff, if one obtains $g(0)=0$, the theory is said to be trivial or Gaussian. Another manifestation of triviality stems from the so-called Landau pole. The effective coupling constant $g(\lambda)$ grows if $\beta(g)$ is positive. It diverges at a finite $\lambda_L$ called the Landau pole. Of course perturbation theory is no more reliable at this point. The singularity disappears for $g(0)\longrightarrow 0$ thus implying triviality. In our context, we use the logarithmic energy scale $\mu$ and we say that the theory is trivial if
\begin{equation}
    \lim\limits_{\mu_{\max}\rightarrow +\infty} f_4(\mu_{\max})=0\;,
\end{equation}
while keeping the bare value $f_4(0)$ fixed. Note that the limit $\mu_{\max}\rightarrow +\infty$, i.e. $\alpha_0\rightarrow 0$ corresponds to removing the UV-cutoff.

\paragraph{} Now we turn to the triviality of the pure mean-field $O(N)$ $\varphi^4_4$-theory. First we prove the existence of solutions of (\ref{MFE_O(N)}) which vanish in the UV-limit for fixed mean-field boundary conditions. Then we will prove the uniqueness of the mean-field trivial solution. 

\subsection{Existence of smooth trivial solutions of the mean-field FE}
\label{Existence_trivial_solution}

\paragraph{}
We consider the following bare lagrangian without irrelevant terms i.e. $c_{0,n}=0$, $n\geq 6$
\begin{equation}\label{bare_lagrangian_trivial}
    L_{0,N}^\mathcal{V}(\varphi)=\int_\mathcal{V} d^4x \Big( c_{0,2}\varphi^2(x)+c_{0,4}\varphi^4(x)\Big)
\end{equation}
and the following (fixed) mean-field boundary conditions following from (\ref{Ais}), (\ref{Mean_effective_action}), (\ref{definition_f_n}) and (\ref{bare_lagrangian_trivial}):

\begin{equation}\label{BDY_trivial_field}
    f_2(0)=2(2\pi)^4\alpha_0 c_{0,2},\quad f_4(0)=4\pi^2 c_{0,4},\quad f_n(0)=0,\quad n\geq 6\;.
\end{equation}

\paragraph{} A direct consequence of (\ref{BDY_trivial_field}) is

\begin{lemma}\label{factorization_of_f_n}
    For smooth solutions $f_n(\mu)$ of (\ref{MFE_O(N)}) with boundary conditions (\ref{BDY_trivial_field}), we have
    \begin{equation}
        \partial_\mu^l f_n(0)=0\;,\quad n\geq 6,\; 0\leq l\leq\frac{n}{2}-3\;.
    \end{equation}
\end{lemma}

\begin{proof}
See \cite{Kopper2022}.
\end{proof}

From Lemma \ref{factorization_of_f_n}, we can set 
\begin{equation}
    f_n(\mu)=\mu^{\frac{n}{2}-2}g_n(\mu)\;,\quad n\geq 4\;,
\end{equation}
where $g_n(\mu)$ are smooth. We can then rewrite the dynamical system (\ref{MFE_O(N)}) as 

\begin{align}\label{MFE-g_n_O(N)}
\begin{split}
    \mu^2 g_{n+2} &=\dfrac{1}{n+N}\sum_{\substack{n_1+n_2=n+2\\
n_i\geq 4}} g_{n_1}g_{n_2}+\mu\dfrac{1}{n+N}g_n\left(2f_2+1-\dfrac{4}{n}\right) \\
                  &+\dfrac{n-4}{n(n+N)}g_n+\dfrac{2}{n(n+N)}\mu\partial_\mu g_n,\quad n\geq 4\;.
\end{split}
\end{align}

Expanding $f_2$ and $g_n$ as formal Taylor series around $\mu=0$
\begin{equation}\label{power_series_O(N)}
    f_2(\mu)=\sum_{k\geq 0}f_{2,k}\mu^k,\quad g_n(\mu)=\sum_{k\geq 0}g_{n,k}\mu^k\;,
\end{equation}
we get

\begin{align}\label{FE_trivial_field1O(N)}
f_{2,k+1} &=\dfrac{1}{k+1}\left((N+2)g_{4,k}+f_{2,k}-\sum_{\nu=0}^k f_{2,\nu}f_{2,k-\nu}\right), \\
\label{FE_trivial_g_n,k}
\begin{split}
g_{n,k+2} &=-\dfrac{n-4}{n+2k}g_{n,k+1}-\dfrac{2n}{n+2k}\sum_{\nu=0}^{k+1}g_{n,\nu}f_{2,k+1-\nu}-\dfrac{n}{n+2k}\sum_{\substack{n_1+n_2=n+2\\
n_i\geq 4}}\sum_{\nu=0}^{k+2}g_{n_1,\nu}g_{n_2,k+2-\nu} \\ &+\dfrac{n(n+N)}{n+2k}g_{n+2,k}\;.
\end{split}
\end{align}
The first line of (\ref{FE_trivial_field1O(N)}) corresponds to (\ref{MFE-g_n_O(N)}) at $n=2$ while the second and third lines of (\ref{FE_trivial_field1O(N)}) correspond to (\ref{MFE-g_n_O(N)}) for $n\geq 4$. Regularity at $\mu=0$ implies for $n\geq 4$

\begin{align}\label{FE_trivial_field2_1_O(N)}
\dfrac{n-4}{n}g_{n,0} &+\sum_{\substack{n_1+n_2=n+2\\
n_i\geq 4}}g_{n_1,0}g_{n_2,0} =0\;,\\
\label{FE_trivial_field2_2_O(N)}
\dfrac{n-2}{n}g_{n,1} &+2\sum_{\substack{n_1+n_2=n+2\\
n_i\geq 4}}g_{n_1,0}g_{n_2,1}+g_{n,0}\left(2f_{2,0}+1-\dfrac{4}{n}\right)=0\;.
\end{align}

\paragraph{} In \cite{Kopper2022} it was proven for $N=1$
\begin{theorem}[Triviality in pure weakly-coupled mean-field $\varphi^4$-theory]\label{weaklytriviality}
    For boundary conditions (\ref{BDY_trivial_field}) such that
     \begin{equation}
    0\leq c_{0,4}\leq\dfrac{\varepsilon}{2^7\pi^2},\quad\vert c_{0,2}\vert\leq\Lambda_0^2\dfrac{\varepsilon}{2^7\pi^4},\quad 0\leq\varepsilon\leq 10^{-2},\;\Lambda_0^{-2}=\alpha_0\;.
\end{equation}
there exist smooth solutions of (\ref{MFE_O(N)}) $f_n\in C^\infty([0,\mu_{\max}])$ which vanish in the UV-limit, i.e. in the limit $\mu_{\max}\longrightarrow +\infty$.
\end{theorem}

\begin{proof}
    See \cite{Kopper2022}.
\end{proof}

 \paragraph{} The key point of the proof is the construction of a two-point function $f_2(\mu)$ such that the mean-field smooth solutions $f_n(\mu)$ turn out to be trivial. In \cite{Kopper2022}, the ansatz for $f_2(\mu)$ is defined by
\begin{equation}\label{ansatz}
    f_{2}(\mu)=\sum_{n\geq 1}b_n\dfrac{x_n^{n-1}}{1+x_n^n},\quad\forall n\geq 1,\; x_n:=n\mu\;.
\end{equation}

\begin{remark}
The ansatz proposed in (\ref{ansatz}) is not analytic at $\mu=0$.
\end{remark}

By expanding $f_2(\mu)$ as in (\ref{power_series_O(N)}), its Taylor coefficients can be rewritten as

\begin{equation}\label{b_n_relation}
    f_{2,k}=(k+1)^k\sum_{\rho=1}^{k+1}b_{\lbrace \frac{k+1}{\rho}\rbrace}(-1)^{\rho-1}\dfrac{1}{\rho^k}\;,
\end{equation}
where by convention $b_0=0$ and
\begin{equation}
    \left\lbrace \frac{m}{n}\right\rbrace:=\left\{
    \begin{array}{ll}
        \frac{m}{n} & \mbox{if } \frac{m}{n}\in\N \\
        0 & \mbox{otherwise.}
    \end{array}
\right.
\end{equation} 
From (\ref{ansatz})-(\ref{b_n_relation}), $f_{2,0}=b_1$ and $f_{2,1}=2b_2-b_1$ where $f_{2,1}=3 f_{4,0}-f_{2,0}(f_{2,0}-1)$, therefore the values of $b_1$ and $b_2$ are fixed by the free choice of $f_{2,0}$ and $f_{4,0}$. The $b_n$'s, $n\geq 3$ are then uniquely determined by (\ref{FE_trivial_field1O(N)})-(\ref{FE_trivial_field2_1_O(N)}) because of the boundary conditions (\ref{BDY_trivial_field}) and the smoothness condition. From (\ref{b_n_relation}) we also have for $n\geq 1$
\begin{equation}\label{b_n_induction}
    b_{n+1}=\frac{f_{2,n}}{(n+1)^n}-\sum_{\rho=2}^{n+1}b_{\lbrace \frac{n+1}{\rho}\rbrace}(-1)^{\rho-1}\dfrac{1}{\rho^n}\;.
\end{equation}

\paragraph{} The fact that $f_2(\mu)$ is well defined on $[0,\mu_{\max}]$ follows from

\begin{proposition}
Under the assumptions of Theorem \ref{weaklytriviality} and choosing the two-point function $f_2(\mu)$ as in (\ref{ansatz}), the following bounds hold:
\begin{equation}\label{initial_bounds_b_n_2022}
   \vert b_n\vert\leq 4\left(\dfrac{3}{4}\right)^n\varepsilon\;,\quad n\geq 1\;.
\end{equation}
\end{proposition}

\begin{proof}
    See \cite{Kopper2022}.
\end{proof}

\paragraph{} Note  that this result implies $\lim\limits_{\mu_{\max}\rightarrow +\infty} f_2(\mu_{\max})=0$. The triviality follows then from these bounds with the aid of the flow equations. This triviality result in \cite{Kopper2022} is weaker than the triviality statements \cite{Aizenman2021}-\cite{Copin2021} as they do not require any upper bound on the value of the bare coupling constant. The aim of this section is to extend Theorem \ref{weaklytriviality} to arbitrarily large values of the mean field couplings. We will follow the steps in \cite{Kopper2022}: choose the two-point function as in (\ref{ansatz}), derive bounds on $g_{n,k}$ and $f_{2,k}$ for given on $g_{4,0}$ and $f_{2,0}$ and derive bounds on $b_n$ which imply that $f_2(\mu)$ is well-defined on $[0,\mu_{\max}]$. We start proving

\begin{lemma}\label{triviality_lemma_1}
    Let $f_n$ be solutions of (\ref{MFE_O(N)}) which respect the boundary conditions (\ref{BDY_trivial_field}). For given $f_{2,0},f_{4,0}$ we choose $K>1$ sufficiently large such that
\begin{equation}\label{initial_bounds}
    \vert f_{2,0}\vert\leq\dfrac{\sqrt{K}}{4},\quad \vert f_{4,0}\vert=\vert g_{4,0}\vert\leq\dfrac{\sqrt{K}}{32}\;.
\end{equation}
Then
\begin{equation}
    \vert f_{2,1}\vert\leq \dfrac{KN}{2},\quad\vert g_{4,1}\vert\leq \dfrac{K}{32}\;,
\end{equation}
and for $n\geq 6$
\begin{equation}\label{general_bounds_n}
    \vert g_{n,0}\vert\leq\dfrac{K^{\frac{n}{2}-\frac{3}{2}}}{2n^2}\;,\quad \vert g_{n,1}\vert\leq \dfrac{K^{\frac{n}{2}-\frac{3}{2}}}{n^2}\left(1+\dfrac{nK}{2}\right)\;.
\end{equation}
\end{lemma}

\begin{proof}
    From (\ref{FE_trivial_field1O(N)}) we have
\begin{equation}
    \vert f_{2,1}\vert =\vert(N+2)g_{4,0}+f_{2,0}-f_{2,0}^2 \vert\leq \dfrac{KN}{2}\;.
\end{equation}
and
\begin{equation*}
    \vert g_{4,1}\vert=4\vert g_{4,0} f_{2,0}\vert\leq\dfrac{K}{32}\;.
\end{equation*}
We proceed by induction in $n$. For $n\geq 6$, we find from (\ref{FE_trivial_field2_1_O(N)}) and for $K$ large enough
\[\vert g_{n,0}\vert\leq \dfrac{n}{n-4}\dfrac{1}{4}\sum_{\substack{n_1+n_2=n+2\\
n_i\geq 4}}\dfrac{K^{\frac{n}{2}-\frac{3}{2}+1-\frac{3}{2}}}{n_1^2(n+2-n_1)^2}\leq\dfrac{K^{\frac{n}{2}-\frac{3}{2}}}{2n^2}\;.\]
From (\ref{FE_trivial_field2_2_O(N)}), for $n\geq 6$ and choosing $K>4$
\begin{align*}\vert g_{n,1}\vert &\leq \dfrac{2n}{n-2}\dfrac{1}{2}\sum_{\substack{n_1+n_2=n+2\\
n_i\geq 4}}\dfrac{K^{\frac{n}{2}-\frac{3}{2}+1-\frac{3}{2}}}{n_1^2(n+2-n_1)^2}\left(1+\dfrac{n_2K}{2}\right)+\dfrac{n}{n-2}\dfrac{K^{\frac{n}{2}-\frac{3}{2}}}{2n^2}\left(\dfrac{\sqrt{K}}{2}+1-\dfrac{4}{n}\right) \\
&\leq\dfrac{K^{\frac{n}{2}-\frac{3}{2}}}{n^2}\left(1+\dfrac{nK}{2}\right)\;.
\end{align*}
The previous bounds for the sums over $n_1$ can be checked for $n\leq 10$. For $n\geq 12$, we use Lemma \ref{inversesquare} in Appendix \ref{appendix_C}.
\end{proof}

\paragraph{} We define for $n_1,n_2\in 2\N\backslash \lbrace 2\rbrace$ and  $k,\nu\in\N_0$
\begin{equation}
g(n_1,n_2,k,\nu):=\dfrac{\left\vert\frac{n_1}{4}+\nu-3\right\vert!\left\vert\frac{n_2}{4}+k-\nu-1\right\vert!\;}{(k+2-\nu)!\;\nu!}\;,
\end{equation} 
where for $n\in\C\backslash \Z_0^-$ we define $n!\;:=\Gamma(n+1)$  with $\Gamma$ the Gamma function. Furthermore we also extend the definition of the binomial coefficient $\binom{n}{k}$ to $n,k\in\C\backslash \Z_0^-$ such that $n-k\in\C\backslash\Z_0^-$ by
\[\binom{n}{k}=\dfrac{n!}{k!\;(n-k)!}\;.\]

We also define for $l\in (0,1)$, $n_1,n_2\in 2\N\backslash \lbrace 2\rbrace$ and  $k,a,b\in\N_0$
\begin{equation}F(n_1,n_2,k,a,b,l):=\sum_{\nu=a}^{k+2-b}\dfrac{\left\vert\frac{n_1}{4}+\nu-3\right\vert!\left\vert\frac{n_2}{4}+k-\nu-1\right\vert!\;}{[(k+2-\nu)!\;\nu!]^{l}}\;.
\end{equation}

\paragraph{} These two quantities will appear in the proofs of the bounds on $g_{n,k}$. Before establishing bounds on $g_{n,k}$ we establish useful lemmas.

\begin{lemma}\label{bound_sum}
    For $n_1,n_2,k,a,b\in\N_0$ such that 
    \begin{itemize}
        \item $a+b\leq k+2$
        \item $3-\frac{n_1}{4}\leq a$ and $3-\frac{n_2}{4}\leq b$
    \end{itemize}
   we have
    \begin{equation}
        \mathcal{S}(n_1,n_2,k,a,b):=\sum_{\nu=a}^{k+2-b}g(n_1,n_2,k,\nu)\leq\dfrac{1}{\frac{n_1+n_2}{4}+a+b-5}\;\dfrac{\left(\frac{n_1+n_2}{4}+k-3\right)!}{(k+2-(a+b))!}\;.
    \end{equation}
\end{lemma}

\begin{proof}
We will use the following equality, found in Sect.1.10 of \cite{Gould_combinatorialidentities}
        
    \begin{equation}\label{combinatorial}
    \sum_{\nu=0}^m \binom{a+\nu}{\nu}\binom{r+m-\nu}{m-\nu}=\binom{a+r+m+1}{m},\quad a,r\geq 0,\;m\in\mathbb{N}\;.
    \end{equation}
Using $(\ref{combinatorial})$ we get
    \begin{equation}
    \begin{split}
       \mathcal{S}(n_1,n_2,k,a,b) &= \sum_{\nu=0}^{k+2-(a+b)}\dfrac{\left(\frac{n_1}{4}-3+a+\nu\right)!\;\left(\frac{n_2}{4}-3+b+k+2-(a+b)-\nu\right)!}{(\nu+a)!\;(k+2-a-\nu)\;!} \\
       &\leq \sum_{\nu=0}^{k+2-(a+b)}\dfrac{\left(\frac{n_1}{4}-3+a+\nu\right)!\;\left(\frac{n_2}{4}-3+b+k+2-(a+b)-\nu\right)!}{\nu!\;(k+2-(a+b)-\nu)\;!} \\
       &\leq\sum_{\nu=0}^{k+2-(a+b)}\left(\frac{n_1}{4}-3+a\right)!\;\left(\frac{n_2}{4}-3+b\right)!\;\binom{\frac{n_1}{4}+a-3+\nu}{\nu}\\
       &\binom{\frac{n_2}{4}+b-3+k+2-(a+b)-\nu}{k+2-(a+b)-\nu} \\
       &\leq \left(\frac{n_1}{4}-3+a\right)!\;\left(\frac{n_2}{4}-3+b\right)!\;\binom{\frac{n_1+n_2}{4}+k-3}{k+2-(a+b)} \\
       &\leq \dfrac{1}{\frac{n_1+n_2}{4}+a+b-5}\;\dfrac{\left(\frac{n_1+n_2}{4}+k-3\right)!}{(k+2-(a+b))!}\;,
    \end{split}
    \end{equation}
    where we used
    \[a!\; b!\leq (a+b)!\;,\quad a,b\geq 0\;.\]
\end{proof}

A consequence of Lemma \ref{bound_sum} is 
\begin{lemma}\label{useful_lemma}
    Let $n_1,n_2,k,a$ and $l\in (0,1)$ such that
    \begin{itemize}
        \item $2a\leq k+2$
        \item $3-\frac{n_1}{4}\leq a$ and $3-\frac{n_2}{4}\leq a$
    \end{itemize}
    then:
    \begin{equation}
        F(n_1,n_2,k,a,a,l)\leq [a! (k+2-a)!]^{1-l}\dfrac{1}{\frac{n_1+n_2}{4}+2a-5}\;\dfrac{\left(\frac{n_1+n_2}{4}+k-3\right)!}{(k+2-2a)!}\;.
    \end{equation}
\end{lemma}

\begin{proof}
    We use the following identity
    \begin{equation}
        \nu!\;(k+2-\nu)!\leq a!\;(k+2-a)!\;,\quad a\leq\nu\leq k+2-a\;.
    \end{equation}
    Then from Lemma \ref{bound_sum}, it follows that
    \begin{equation}
    \begin{split}
      F(n_1,n_2,k,a,a,l) &\leq [a!\;(k+2-a)!]^{1-l}\mathcal{S}(n_1,n_2,k,a,a) \\
      &\leq [a! (k+2-a)!]^{1-l}\dfrac{1}{\frac{n_1+n_2}{4}+2a-5}\;\dfrac{\left(\frac{n_1+n_2}{4}+k-3\right)!}{(k+2-2a)!}\;.  
    \end{split}
    \end{equation}
\end{proof}

The following proposition shows that the coefficients $g_{n,k}$ and $f_{2,k}$ grow at most as $[k!]^{\frac{3}{4}}$. This will allow us later to show that the function $f_2(\mu)$ and then also the $f_n(\mu)$ are well-defined for $\mu\in [0,\mu_{\max}]$, see (\ref{b_n_inequality}) and Propositions \ref{b_n_bounds}-\ref{vanishing_solutions_in_the_UV}.

\begin{proposition}\label{triviality_prop}
    Under the same assumptions as in Lemma \ref{triviality_lemma_1}, we have for $n\geq 4,\;k\geq 2,\;N\geq 1$, 
\begin{equation}\label{prop_bound_g_n,k_f_2,k}
    \vert g_{n,k}\vert\leq N^{\frac{n}{2}+k-2}K^{\frac{n}{2}+k-\frac{3}{2}}\left\vert\dfrac{n}{4}+k-3\right\vert!\;\dfrac{1}{(k!)^{\frac{1}{4}}}\;,\quad \vert f_{2,k}\vert\leq N^{k+1} K^{k+\frac{1}{2}}\dfrac{\vert k-3\vert!}{(\vert k-1\vert!)^{\frac{1}{4}}}\;.
\end{equation}
\end{proposition}

\begin{proof}
We proceed by induction going up in $M=n+k$. For a fixed value of $M=n+k$ we go up in $k$. To initialize the induction, the bounds for $k\leq 1$ for $g_{n,k}$ and $f_{2,k}$ follow from  Lemma \ref{triviality_lemma_1}.

We will first bound $g_{n,k+2}$. Using (\ref{FE_trivial_g_n,k}) we can then bound $g_{n,k+2}$, knowing the bounds on the terms appearing on the r.h.s. We proceed term by term.
\begin{enumerate}
    \item We see that for the cases where $\frac{n}{4}+k-3< 0$, namely $(n=4,k=0,1)$, $(n=6,k=0,1)$, $(n=8,k=0)$ and $(n=10,k=0)$, the bounds follow from Lemma \ref{triviality_lemma_1}.
    \item We look at the different terms on the r.h.s of (\ref{FE_trivial_g_n,k}).
    \begin{itemize}
        \item First term: this term vanishes for $n=4$. For $n\geq 6$ and $k=0$ we use the bounds in Lemma \ref{triviality_lemma_1} to obtain
            \begin{equation}
            \begin{split}
                \dfrac{n-4}{n+2}\vert g_{n,1}\vert &\leq N^{\frac{n}{2}-1}\dfrac{n-4}{n+2}\dfrac{K^{\frac{n}{2}-\frac{3}{2}}}{n^2}\left(1+\dfrac{nK}{2}\right)  \\
                &\leq N^{\frac{n}{2}-1}\dfrac{K^{\frac{n}{2}-\frac{3}{2}+1}}{n} \\
                &\leq N^{\frac{n}{2}}K^{\frac{n}{2}+2-\frac{3}{2}}\left(\frac{n}{4}-1\right)!\;\frac{1}{K}\;.
            \end{split}
            \end{equation}

    Then for $n\geq 6$ and $k\geq 1$, using the induction hypothesis we obtain
        
        \begin{equation}\label{firstterm}
        \begin{split}
          &\dfrac{n-4}{n+2k}\vert g_{n,k+1}\vert \leq N^{\frac{n}{2}+k-1}\dfrac{n-4}{n+2k}K^{\frac{n}{2}+k-\frac{1}{2}}\left\vert\dfrac{n}{4}+k-2\right\vert!\;\dfrac{1}{[(k+1)!]^{\frac{1}{4}}} \\
            &\leq N^{\frac{n}{2}+k-1}K^{\frac{n}{2}+k-\frac{1}{2}}\left(\dfrac{n}{4}+k-1\right)!\;\dfrac{(k+2)^{\frac{1}{4}}}{[(k+2)!]^{\frac{1}{4}}}\;\dfrac{n-4}{n+2k}\dfrac{4}{(n+4k-4)} \\
            &\leq N^{\frac{n}{2}+k-1}K^{\frac{n}{2}+k+2-\frac{3}{2}}\left(\dfrac{n}{4}+k-1\right)!\;\dfrac{1}{[(k+2)!]^{\frac{1}{4}}}\; \dfrac{1}{K}\;.  
        \end{split}
        \end{equation}
        \item Second term: We can rewrite the second term as follows
        \begin{equation}\label{second_term_rewritten}
        \begin{split}
            \dfrac{2n}{n+2k}\sum_{\nu=0}^{k+1}\vert g_{n,\nu}f_{2,k+1-\nu}\vert &= \dfrac{2n}{n+2k}\Big(\vert g_{n,0}f_{2,k+1}\vert+\vert g_{n,1}f_{2,k}\vert+\vert g_{n,k+1}f_{2,0}\vert \\
            &+\vert g_{n,k}f_{2,1}\vert 
            + \vert g_{n,k-1}f_{2,2}\vert + \sum_{\nu=2}^{k-2}\vert g_{n,\nu}f_{2,k+1-\nu}\vert\Big)\;.
        \end{split}
        \end{equation}
        
    For the terms with $\nu\leq 1$, we use the bounds in Lemma \ref{triviality_lemma_1} to get:
    \begin{itemize}
        \item $\nu=0$:

    \begin{equation}\label{second_term_nu=0}
    \begin{split}
     \dfrac{2n}{n+2k}\vert g_{n,0}f_{2,k+1}\vert &\leq \dfrac{2n}{n+2k}N^{k+2} \dfrac{K^{\frac{n}{2}-\frac{3}{2}}}{2n^2} K^{k+\frac{3}{2}}\dfrac{\vert k-2\vert !}{(k!)^{\frac{1}{4}}} \\
     &\leq N^{\frac{n}{2}+k}K^{\frac{n}{2}+k+\frac{1}{2}}\left(\frac{n}{4}+k-1\right)!\;\dfrac{1}{[(k+2)!]^{\frac{1}{4}}}\;\dfrac{1}{\sqrt{K}}\;,   
    \end{split}
    \end{equation}
        where we used 
        \begin{equation}\label{inequality_2_1}
           \dfrac{((k+1)(k+2))^{\frac{1}{4}}\vert k-2\vert !}{n^2}\leq \left(\frac{n}{4}+k-1\right)!\;,\quad n\geq4,k\geq 0\;. 
        \end{equation}
         
        \item $\nu=1$:
        \begin{equation}\label{second_term_nu=1}
           \begin{split}
            \dfrac{2n}{n+2k}\vert g_{n,1}f_{2,k}\vert &\leq \dfrac{2n}{n+2k}N^{k+1}\dfrac{3K^{\frac{n}{2}-\frac{1}{2}}}{4n}K^{k+\frac{1}{2}}\dfrac{\vert k-3\vert !}{(\vert k-1\vert !)^{\frac{1}{4}}} \\
            &\leq N^{\frac{n}{2}+k}K^{\frac{n}{2}+k+\frac{1}{2}}\left(\frac{n}{4}+k-1\right)!\;\dfrac{1}{[(k+2)!]^{\frac{1}{4}}}\;\dfrac{6}{\sqrt{K}}\;,   
           \end{split} 
        \end{equation}
    
        where we used the following inequalities, valid for $n\geq 4,\;k\geq 0$,
        \begin{equation}
            \dfrac{K^{\frac{n}{2}-\frac{3}{2}}}{n^2}\left(1+\dfrac{nK}{2}\right)\leq \dfrac{3K^{\frac{n}{2}-\frac{1}{2}}}{4n},\quad \dfrac{\vert k-3\vert !}{4n (\vert k-1\vert !)^{\frac{1}{4}}}\leq \left(\frac{n}{4}+k-1\right)!\;\dfrac{1}{[(k+2)!]^{\frac{1}{4}}}\;.
        \end{equation}

        From Lemma \ref{triviality_lemma_1} and (\ref{FE_trivial_field1O(N)}), we have
        \begin{equation}
            \vert f_{2,2}\vert\leq\dfrac{1}{2}\left((N+2)\vert g_{4,1}\vert+\vert f_{2,1}\vert+2\vert f_{2,0} f_{2,1}\vert\right)\leq N K\sqrt{K}\;.
        \end{equation}
        
        \item $k-1\leq \nu\leq k+1$: it is clear that $\vert g_{n,\nu} f_{2,k+1-\nu}\vert$ are bounded by 
        \begin{equation}\label{second_term_k-1<=nu}
           N^{\frac{n}{2}+k} K^{\frac{n}{2}+k+\frac{1}{2}}\left(\dfrac{n}{4}+k-1\right)!\; \dfrac{1}{[(k+2)!]^{\frac{1}{4}}}\; \dfrac{C}{\sqrt{K}}\;,  
        \end{equation}
       
        where $C$ is a constant which does not depend on $n,k$.
        
        \item $2\leq \nu \leq k-2$ We now bound the remaining sum in (\ref{second_term_rewritten}). Note that this sum is non-zero only if $k\geq 4$, so we assume $k\geq 4$ from now on. The remaining sum is bounded by
        \begin{equation}\label{second_term_intermediary_sum}
          N^{\frac{n}{2}+k}K^{\frac{n}{2}+k+\frac{1}{2}}\dfrac{n}{2\sqrt{K}(n+2k)}F\left(n,4,k-2,2,2,\frac{1}{4}\right)\;.  
        \end{equation}

    Then from Lemma \ref{useful_lemma} we have
    \begin{equation}\label{second_term_nu>=2}
    \begin{split}
     \dfrac{n}{n+2k}F\left(n,4,k-2,2,2,\frac{1}{4}\right) &\leq
     \dfrac{n}{n+2k} [(k-2)!]^{\frac{3}{4}}\dfrac{4}{n}\dfrac{\left(\frac{n}{4}+k-4\right)!}{(k-4)!} \\
     &\leq 4 \;\dfrac{\left(\frac{n}{4}+k-1\right)!}{[(k+2)!]^{\frac{1}{4}}}\dfrac{\Big(k(k+1)(k+2)\Big)^{\frac{1}{4}}}{n+2k} \\
     &\leq 4 \;\dfrac{\left(\frac{n}{4}+k-1\right)!}{[(k+2)!]^{\frac{1}{4}}}\;.   
    \end{split}
    \end{equation}
    \end{itemize}
    
    Summing (\ref{second_term_nu=0}), (\ref{second_term_nu=1}), (\ref{second_term_k-1<=nu}) and (\ref{second_term_nu>=2}) we find
    \begin{equation}\label{secondterm}
        \dfrac{2n}{n+2k}\sum_{\nu=0}^{k+1}\vert g_{n,\nu}f_{2,k+1-\nu}\vert\leq N^{\frac{n}{2}+k}K^{\frac{n}{2}+k+\frac{1}{2}}\left(\frac{n}{4}+k-1\right)!\;\dfrac{1}{[(k+2)!]^{\frac{1}{4}}}\;\dfrac{C_2}{\sqrt{K}}\;,
    \end{equation}
    where $C_2>0$ is a suitable constant.
    
    \item Third term is bounded by
        \begin{equation}\label{horrible_sum}
        \begin{split}
          &\dfrac{n}{n+2k}\sum_{\substack{n_1+n_2=n+2\\
n_i\geq 4}}\sum_{\nu=0}^{k+2}\vert g_{n_1,\nu}g_{n_2,k+2-\nu}\vert \leq \dfrac{2n}{n+2k}I+N^{\frac{n}{2}+k-1}\dfrac{2K^{\frac{n}{2}+k+\frac{1}{2}}n}{\sqrt{K}(n+2k)}\times \\
&\Bigg[\sum_{\substack{4\leq n_1\leq 10\\
n_1\leq n_2\\ n_1+n_2=n+2}}F\left(n_1,n_2,k,2,2,\frac{1}{4}\right)+\sum_{\substack{12\leq n_1\leq n_2\\ n_1+n_2=n+2}}F\left(n_1,n_2,k,0,0,\frac{1}{4}\right)\Bigg]\;,
        \end{split}
        \end{equation}
where we define

\begin{equation}
    I:=\sum_{\substack{4\leq n_1\leq 10 \\ n_1\leq n_2 \\ n_1+n_2=n+2}}\vert g_{n_1,0}g_{n_2,k+2}\vert+\vert g_{n_1,1}g_{n_2,k+1}\vert+\vert g_{n_1,k+2}g_{n_2,0}\vert+\vert g_{n_1,k+1}g_{n_2,1}\vert\;.
\end{equation}

The first and the second term in the r.h.s. of (\ref{horrible_sum}) contains a finite number of terms which does not depend on $n$. Using the bounds from Lemma \ref{triviality_lemma_1} and the induction hypothesis, it is not too hard to prove that
\begin{equation}
    I\leq N^{\frac{n}{2}+k}K^{\frac{n}{2}+k+\frac{1}{2}}\left(\frac{n}{4}+k-1\right)!\dfrac{1}{[(k+2)!]^{\frac{1}{4}}}\dfrac{C'}{\sqrt{K}}\;,
\end{equation}
where $C'$ is a constant which does not depend on $n,k$. Then we use Lemma \ref{useful_lemma} to obtain respectively

\begin{equation}
    \dfrac{n}{n+2k}\sum_{\substack{4\leq n_1\leq 10\\
n_1\leq n_2\\ n_1+n_2=n+2}}F\left(n_1,n_2,k,2,2,\frac{1}{4}\right)\leq \dfrac{4n}{n+2k}\dfrac{[2k!]^{\frac{3}{4}}}{\frac{n}{4}-\frac{1}{2}}\dfrac{\left(\frac{n}{4}+k-\frac{5}{2}\right)!}{(k-2)!}
\end{equation}  

and\footnote{$12\leq n_1\leq n_2$ implies  $n\geq 22$, so that the denominators are non-zero.}

\begin{equation}
    \dfrac{n}{n+2k}\sum_{\substack{12\leq n_1\leq n_2\\ n_1+n_2=n+2}}F\left(n_1,n_2,k,0,0,\frac{1}{4}\right)\leq \frac{n}{2}\dfrac{2}{\frac{n}{4}-\frac{9}{2}}\dfrac{\left(\frac{n}{4}+k-\frac{5}{2}\right)!}{[(k+2)!]^{\frac{1}{4}}}\;.
\end{equation}

Using the inequality
\[\left(m+\frac{1}{2}\right)!\leq 2m!\;\sqrt{m+\frac{1}{2}},\quad m\in\N\;,\]
we obtain finally

\begin{equation}
\begin{split}
\dfrac{n}{n+2k}\sum_{\substack{4\leq n_1\leq 10\\
n_1\leq n_2\\ n_1+n_2=n+2}}F\left(n_1,n_2,k,2,2,\frac{1}{4}\right) &\leq \dfrac{8n}{n+2k}\dfrac{2^{\frac{3}{4}}}{\frac{n}{4}-\frac{1}{2}}\dfrac{\left(\frac{n}{4}+k-1\right)!\;\sqrt{\frac{n}{4}+k-\frac{5}{2}}}{[k!]^{\frac{1}{4}}} \\
&\leq C''\dfrac{\left(\frac{n}{4}+k-1\right)!}{[(k+2)!]^{\frac{1}{4}}}\;,
\end{split}
\end{equation}
where $C''$ does not depend on $n,k$.

\paragraph{} Summing over all contributions, we get the bound
\begin{equation}\label{thirdterm}
    \dfrac{n}{n+2k}\sum_{\substack{n_1+n_2=n+2\\
n_i\geq 4}}\sum_{\nu=0}^{k+2}\vert g_{n_1,\nu}g_{n_2,k+2-\nu}\vert\leq N^{\frac{n}{2}+k}K^{\frac{n}{2}+k+\frac{1}{2}}\left(\frac{n}{4}+k-1\right)!\;\dfrac{1}{[(k+2)!]^{\frac{1}{4}}}\dfrac{C_3}{\sqrt{K}}\;,
\end{equation}
for some finite positive constant $C_3$.

     \item Fourth term: First, for $k\leq 1$, we use the bounds in Lemma \ref{triviality_lemma_1} to obtain
     \begin{itemize}
         \item $k=0$
         \begin{equation}
             \dfrac{n(n+N)}{n}\vert g_{n+2,0}\vert\leq N^{\frac{n}{2}}K^{\frac{n}{2}+\frac{1}{2}}\left(\frac{n}{4}-1\right)!\;\dfrac{1}{K}\;.
         \end{equation}
         \item $k=1$
         \begin{equation}
             \dfrac{n(n+N)}{n+2}\vert g_{n+2,1}\vert\leq N^{\frac{n}{2}+1}K^{\frac{n}{2}+\frac{3}{2}}\left(\frac{n}{4}\right)!\;\dfrac{1}{K}\;.
         \end{equation}
     \end{itemize}
     For $k\geq 2$ we get
     \begin{align}\label{fourthterm}
     \begin{split}
        &\dfrac{n(n+N)}{n+2k}\vert g_{n+2,k}\vert \leq \dfrac{n(n+N)}{n+2k}N^{\frac{n}{2}+k-1} K^{\frac{n}{2}+k-\frac{1}{2}}\left(\dfrac{n}{4}+k+\dfrac{1}{2}-3\right)!\;\dfrac{1}{(k!)^{\frac{1}{4}}} \\
    &\leq \dfrac{N^{\frac{n}{2}+k}\;K^{\frac{n}{2}+k+\frac{1}{2}}\left(\frac{n}{4}+k-1\right)!}{[(k+2)!]^{\frac{1}{4}}}\;\dfrac{2n(n+1)\sqrt{\frac{n}{4}+k-\frac{5}{2}}[(k+1)(k+2)]^{\frac{1}{4}}}{(n+2k)(n+4k-4)(n+4k-8)K} \\
    &\leq N^{\frac{n}{2}+k}K^{\frac{n}{2}+k+\frac{1}{2}}\left(\dfrac{n}{4}+k-1\right)!\;\dfrac{1}{[(k+2)!]^{\frac{1}{4}}}\;\dfrac{1}{K}\;. 
     \end{split}
    \end{align}
    \end{itemize}
Summing (\ref{firstterm}), (\ref{secondterm}), (\ref{thirdterm}) and (\ref{fourthterm}), we obtain
\begin{equation}
\begin{split}
 \vert g_{n,k+2}\vert &\leq \left[\frac{1}{K}+\frac{C_2+C_3}{\sqrt{K}}+\frac{1}{K}\right]N^{\frac{n}{2}+k}\dfrac{K^{\frac{n}{2}+k+\frac{1}{2}}}{[(k+2)!]^{\frac{1}{4}}}\left(\dfrac{n}{4}+k-1\right)!\\
 &\leq N^{\frac{n}{2}+k}\dfrac{K^{\frac{n}{2}+k+\frac{1}{2}}}{[(k+2)!]^{\frac{1}{4}}}\left(\dfrac{n}{4}+k-1\right)!\;.    
\end{split}  
\end{equation}

\paragraph{} We will now bound the $f_{2,k}$.  One can check term by term that the bound (\ref{prop_bound_g_n,k_f_2,k}) is true for $k\leq 6$. For $K$ large enough, and for $k>6$ we obtain from (\ref{FE_trivial_field1O(N)}) by induction

\begin{equation}\label{f_2,k_first_bound}
    \begin{split}
 \vert f_{2,k+1}\vert &\leq\dfrac{1}{k+1}\Bigg((N+2)N^k\dfrac{K^{k+\frac{1}{2}}(k-2)!}{(k!)^{\frac{1}{4}}}+N^{k}K^{k+\frac{1}{2}}\dfrac{(k-3)!}{[(k-1)!]^{\frac{1}{4}}} \\
 &+N^{k+1}K^{k+1}\sum_{\nu=0}^k \dfrac{\vert\nu-3\vert! \vert k-\nu-3\vert !}{[\vert\nu-1\vert!\;\vert k-\nu-1\vert!]^{\frac{1}{4}}}\Bigg)\;.       
    \end{split}
\end{equation}

We will bound each term in the r.h.s. of (\ref{f_2,k_first_bound}). We will again proceed term by term.

\begin{itemize}
    \item First term: we trivially have
    \begin{align*}
        \dfrac{1}{(k+1)}\dfrac{K^{k+\frac{1}{2}}(k-2)!}{(k!)^{\frac{1}{4}}} &\leq\dfrac{1}{K(k+1)} \dfrac{K^{k+1+\frac{1}{2}}(k-2)!}{(k!)^{\frac{1}{4}}}\;. 
    \end{align*}
    \item Second term: obviously
    \begin{align*}
        \dfrac{K^{k+\frac{1}{2}}}{k+1}\dfrac{(k-3)!}{((k-1)!)^{\frac{1}{4}}}\leq\dfrac{k^{\frac{1}{4}}}{K(k+1)(k-2)} \dfrac{K^{k+1+\frac{1}{2}}(k-2)!}{(k!)^{\frac{1}{4}}}\;.
    \end{align*}
    \item Third term: we note that $[(\nu-1)!\;( k-\nu-1)!]^{\frac{3}{4}}\leq [(k-2)!]^{\frac{3}{4}}\leq\left(\frac{k!}{k(k-1)}\right)^{\frac{3}{4}}$ for $3\leq\nu\leq k-3$. Then the terms summed over $3\leq\nu\leq k-3$ are bounded by
    \begin{equation}
        \begin{split}
       \dfrac{(k!)^{\frac{3}{4}}}{(k+1)(k(k-1))^{\frac{3}{4}}}\sum_{\nu=3}^{k-3} \dfrac{(\nu-3)!\; ( k-\nu-3) !}{(\nu-1)!\;( k-\nu-1)!} &\leq \frac{(k!)^{\frac{3}{4}}}{(k-5)(k+1)(k(k-1))^{\frac{3}{4}}} \\
       &\leq \dfrac{(k-2)!}{(k!)^{\frac{1}{4}}}\;.     
        \end{split}
    \end{equation}

The remaining terms are symmetric under $\nu'\mapsto k-\nu'$ so that we restrict ourselves to $0\leq\nu\leq 2$. For $k>6$, they are bounded by
\begin{enumerate}
    \item $\nu=0$:
    \[\dfrac{6(k-3)!}{((k-1)!)^{\frac{1}{4}}}\leq 6\dfrac{(k-2)!}{(k!)^{\frac{1}{4}}}\;,\]
    \item $\nu=1$:
    \[\dfrac{4(k-4)!}{((k-2)!)^{\frac{1}{4}}}\leq 4\dfrac{(k-2)!}{(k!)^{\frac{1}{4}}}\;,\]
    \item $\nu=2$:
    \[\dfrac{(k-5)!}{((k-3)!)^{\frac{1}{4}}}\leq\dfrac{(k-2)!}{(k!)^{\frac{1}{4}}}\;,\]
\end{enumerate}
Then by summing we get
\[\dfrac{1}{k+1}\sum_{\nu=0}^k \dfrac{\vert\nu-3\vert! \vert k-\nu-3\vert !}{(\vert\nu-1\vert!\;\vert k-\nu-1\vert!)^{\frac{1}{4}}}\leq\left(1+\frac{22}{7}\right)\dfrac{(k-2)!}{(k!)^{\frac{1}{4}}}\;.\]
\end{itemize}
Altogether we find
\begin{equation}
    \vert f_{2,k+1}\vert\leq \left(\frac{1}{K}+\frac{1}{K}+\frac{5}{\sqrt{K}}\right)N^{k+2}\dfrac{K^{k+1+\frac{1}{2}}(k-2)!}{(k!)^{\frac{1}{4}}}\leq N^{k+2}\dfrac{K^{k+1+\frac{1}{2}}(k-2)!}{(k!)^{\frac{1}{4}}}\;.
\end{equation}
 
This ends the proof.
\end{enumerate}
\end{proof}

\paragraph{} Using the bound (\ref{prop_bound_g_n,k_f_2,k}) in (\ref{b_n_induction}) we have for $n\geq 1$
\begin{equation}\label{b_n_inequality}
    \vert b_{n+1}\vert\leq c_{n,N}+\sum_{\rho=2}^{n+1}\vert b_{\lbrace \frac{n+1}{\rho}\rbrace}\vert\dfrac{1}{\rho^n}\;,\quad c_{n,N}:=N^{n+1}K^{n+\frac{1}{2}}\dfrac{\vert n-3\vert!}{(\vert n-1\vert!)^{\frac{1}{4}}(n+1)^n}\;.
\end{equation}
Note that this bound is sharper than the one obtained in (\ref{initial_bounds_b_n_2022}) and in \cite{Kopper2022} due to the factor $(\frac{1}{\vert n-1\vert!})^{\frac{1}{4}}$ .We set $C_N:=\sum_{n\geq 0}c_{n,N}<+\infty$. Now we establish bounds on the $b_n$'s. We prove the following

\begin{proposition}\label{b_n_bounds}
   There exists $C(N,K)>0$ such that 
   \begin{equation}\label{bound_b_n}
       \vert b_n\vert\leq C(N,K)\frac{n^2}{2^n},\quad n\geq 1\;.
   \end{equation}
\end{proposition}

\begin{proof}
    The proof is done by induction in $n\in\N$. For $n=1$, $b_1=f_{2,0}$ and the bound is obtained choosing  any constant $C(N,K)\geq\frac{N\sqrt{K}}{2}$. For $n\in\N$, we use (\ref{b_n_inequality}) to obtain
    \begin{equation}\label{b_n_bounds_induction1}
    \begin{split}
      \vert b_{n+1}\vert &\leq c_{n,N}+\sum_{\rho=2}^{n+1} C(N,K)\left(\frac{n+1}{\rho}\right)^2\dfrac{1}{2^{\frac{n+1}{\rho}}\rho^n} \\
      &\leq c_{n,N}+C(N,K)\dfrac{(n+1)^2}{2^{n+3}}+(n+1)^2 C(N,K)\sum_{\rho=3}^{n+1} \dfrac{1}{2^{\frac{n+1}{\rho}}\rho^{n+2}}\;.  
    \end{split}
    \end{equation}
    There exists a constant $\Tilde{C}(N,K)>0$ such that 
    \begin{equation}\label{b_n_bounds_induction2}
      c_{m,N}\leq \Tilde{C}(N,K)\frac{(m+1)^2}{2^{m+3}},\quad m\geq 1\;.  
    \end{equation}
    Moreover we have
    \begin{equation}\label{b_n_bounds_induction3}
        \begin{split}
            \sum_{\rho=3}^{n+1} \dfrac{1}{2^{\frac{n+1}{\rho}}\rho^{n+2}} \leq \sum_{\rho=3}^{n+1} \dfrac{1}{\rho^{n+2}}\leq \int_2^{n+1}\dfrac{dx}{x^{n+2}}=\dfrac{1}{2^{n+1}(n+1)}-\dfrac{1}{(n+1)^{n+2}}\leq \dfrac{1}{2^{n+2}}\;.
        \end{split}
    \end{equation}
    From (\ref{b_n_bounds_induction2}) and (\ref{b_n_bounds_induction3}), we get
    \begin{equation}
        \vert b_{n+1}\vert\leq \left[\frac{\Tilde{C}(N,K)}{4}+\frac{C(N,K)}{4}+\frac{C(N,K)}{2}\right]\frac{(n+1)^2}{2^{n+1}}\leq C(N,K)\frac{(n+1)^2}{2^{n+1}}\;,
    \end{equation}
    if we choose $C(N,K)\geq \max(\tilde{C}(N,K),\frac{\sqrt{K}}{2})$.
\end{proof}

We also prove the following 
\begin{proposition}\label{derivativesf_2}
    For $l\geq 0$, $n>l+1$, we have
\begin{equation}
    \left\vert \partial_\mu^l \frac{x_n^{n-1}}{1+x_n^n}\right\vert\leq \frac{n^{n+l-1}\mu^{n-l-1}}{1+x_n^n}\mathcal{C}_l\;,\quad\mu\in[0,\mu_{\max}]\;,
\end{equation}
where the integers $\mathcal{C}_l$ are defined by $\mathcal{C}_0=1$ and 
\begin{equation}
    \mathcal{C}_{l+1}=1+\sum_{j=0}^l\binom{l+1}{j}\mathcal{C}_j\leq 4^{l+1} (l+1)!\;.
\end{equation}
\end{proposition}

\begin{proof}
    We prove the proposition by induction in $l\geq 0$.
\begin{itemize}
    \item The case $l=0$ is obvious.
    \item We use the following formula
    \begin{equation}\label{derivative_of_f/g_formula}
    \left(\dfrac{f}{g}\right)^{(l)}=\frac{1}{g}\left[f^{(l)}-l!\;\sum_{j=1}^{l}\frac{g^{(l+1-j)}}{(l+1-j)!}\frac{1}{(j-1)!}\left(\dfrac{f}{g}\right)^{(j-1)}\right]\;,
\end{equation}
for $f,g$ smooth functions with $g>0$. See Appendix \ref{l-th_derivative_f/g} for a proof. We have for $n>l+2$ and $\mu\geq 0$
\begin{align*}
    \left\vert \partial_\mu^{l+1} \frac{x_n^{n-1}}{1+x_n^n}\right\vert &\leq \dfrac{1}{1+x_n^n}\Bigg[n^{n-1}\prod_{m=0}^{l}(n-1-m)\;\mu^{n-2-l} \\
    &+(l+1)!\sum_{j=1}^{l+1}\dfrac{1}{(l+2-j)!\;(j-1)!}n^n\prod_{m=0}^{l+1-j}(n-m)\;\mu^{n-l-2+j}\frac{n^{n+j-2}\mu^{n-j}}{1+x_n^n}\mathcal{C}_{j-1}\Bigg] \\
    &\leq \frac{1}{1+x_n^n}\left[n^{n-1}n^{l+1}\mu^{n-l-2}+\mu^{n-l-2}\sum_{j=1}^{l+1}\binom{l+1}{j-1}n^{l+2-j}n^{n+j-2}C_{j-1}\right] \\
    &\leq \dfrac{n^{n+l}\mu^{n-l-2}}{1+x_n^n}\left[1+\sum_{j=0}^l\binom{l+1}{j}\mathcal{C}_j\right]\leq\dfrac{n^{n+l}\mu^{n-l-2}}{1+x_n^n}\mathcal{C}_{l+1}\;.
\end{align*}
\end{itemize}
The bound on $C_l$ can be straightforwardly proven by induction in $l$.
\end{proof}

\paragraph{}

Now we can prove the last result concerning the behaviour of the mean-field smooth solutions $f_n(\mu)$ in the UV-limit.

\begin{proposition}\label{vanishing_solutions_in_the_UV}
\begin{itemize}
    \item $f_2(\mu)$ is well defined on $[0,\mu_{\max}]$ and
\begin{equation}\label{vanishing_f_2}
    \quad\lim\limits_{\mu_{\max}\longrightarrow +\infty}\partial_\mu^l f_2(\mu_{\max})=0,\quad l\geq 0\;.
    \end{equation}
    \item The functions $\partial_\mu^l f_n(\mu)$, $l\geq 0, n\geq 4$ are well defined on $[0,\mu_{\max}]$ and
\begin{equation}\label{vanishing_f_n}
    \lim\limits_{\mu_{\max}\longrightarrow +\infty}\partial_\mu^l f_n(\mu_{\max})=0,\quad n\geq 4,\; l\geq 0\;.
\end{equation}
\end{itemize}  
\end{proposition}

\begin{proof}
As a consequence of Proposition \ref{derivativesf_2}, for $l\geq 0$, $n>l+1$ and $\mu\geq 0$
\begin{equation}
    \left\vert \partial_\mu^l \frac{x_n^{n-1}}{1+x_n^n}\right\vert\leq n^{2l}\dfrac{x_n^{n-l-1}}{1+x_n^n}C_l\;.
\end{equation}
For $0\leq l<n-1$, the function $g(\mu)=\dfrac{(n\mu)^{n-l-1}}{1+(n\mu)^n}$ defined on $\R_+$ reaches its maximum at $\tilde{\mu}=\frac{1}{n}\left(\frac{n}{l+1}-1\right)^\frac{1}{n}$ and 
\begin{equation}
   \Vert g\Vert_\infty = g(\tilde{\mu})=\frac{l+1}{n}\left(\frac{n}{l+1}-1\right)^{1-\frac{l+1}{n}}\leq 1\;.
\end{equation}
 Proposition \ref{derivativesf_2} and the bounds (\ref{bound_b_n}) imply the uniform convergence of $f_2$ and its derivatives on $[0,\mu_{\max}]$. As a result we have
\begin{equation}
\forall l\geq 0,\;\forall m\geq 1,\quad\lim\limits_{\mu\longrightarrow +\infty} \partial_\mu^l\dfrac{x_m^{m-1}}{1+x_m^m}=0\quad\Longrightarrow\quad
\forall l\geq 0, \quad\lim\limits_{\mu_{\max}\longrightarrow +\infty}\partial_\mu^l f_2(\mu_{\max})=0\;.
\end{equation}   

\paragraph{} The proof of the second statement is done by induction in $n+l$, going up in $n$. We can then easily check the case $n=4$ using (\ref{MFE_O(N)}) and (\ref{vanishing_f_2}). For $n\geq 4$, we obtain by differentiating (\ref{MFE_O(N)}) $l$ times w.r.t. $\mu$

\begin{equation}
    \partial_\mu^l f_{n+2}=\dfrac{2}{n(n+N)}\partial_\mu^{l+1} f_n+\dfrac{n-4}{n(n+N)}\partial_\mu^l f_n+\dfrac{1}{n+N}\sum_{n_1+n_2=n+2}\quad \sum_{l_1+l_2=l}\binom{l}{l_1}\partial_\mu^{l_1}f_{n_1}\partial_\mu^{l_2}f_{n_2}\;.
\end{equation}

Using the induction hypothesis, $\partial_\mu^l f_{n+2}(\mu)$ are well defined on $[0,\mu_{\max}]$ and they vanish when $\mu_{\max}\longrightarrow +\infty$.
\end{proof}

\paragraph{} Collecting our findings, we can now state our existence result.

\begin{theorem}\label{theorem_triviality_massless}
    Consider a $O(N)$ vector model $\varphi^4_4$-theory of bare interaction lagrangian (\ref{bare_lagrangian_trivial}). Let $f_n(\mu)$ be smooth  solutions of the mean-field flow equations (\ref{MFE_O(N)}) and we consider the mean-field boundary conditions (\ref{BDY_trivial_field}) with
    \begin{equation}
        0<c_{0,4}<+\infty,\quad \vert c_{0,2}\vert<+\infty\;.
    \end{equation}
    There exist smooth solutions of (\ref{MFE_O(N)}) $f_n(\mu)\in C^\infty([0,\mu_{\max}])$ such that they vanish in the UV-limit, i.e.
    \begin{equation}
        \lim\limits_{\mu_{\max}\longrightarrow +\infty} f_n(\mu_{\max})=0,\quad n\geq 2\;.
    \end{equation}
\end{theorem}

\begin{proof} 

 For $c_{0,4}<+\infty$ and $\vert c_{0,2} \vert<+\infty$ fixed, there exists $K$  such that
\begin{equation}\label{modifiedbdyconditions}
    0\leq c_{0,4}\leq\dfrac{\sqrt{K}}{2^7\pi^2},\quad\vert c_{0,2}\vert\leq\Lambda_0^2\dfrac{\sqrt{K}}{2^7\pi^4},\quad\Lambda_0^{-2}=\alpha_0\;,
\end{equation}
Using Lemma \ref{triviality_lemma_1} and Proposition \ref{b_n_bounds}, we choose $K$ such that the bounds (\ref{prop_bound_g_n,k_f_2,k}) hold. Then we choose a smooth two-point function $f_2(\mu)$ as in (\ref{ansatz}). From Proposition \ref{b_n_bounds}, Proposition \ref{derivativesf_2} and Proposition \ref{vanishing_solutions_in_the_UV}, the obtained smooth solutions $f_n(\mu)$ vanish in the UV-limit.

\end{proof}

\paragraph{} We finish this section with a few remarks

\begin{remarks}
    \begin{itemize}
        \item The limit $\mu_{\max}\rightarrow +\infty$ or $\alpha_0\rightarrow 0$ is equivalent to removing the UV-cutoff. In statistical mechanics we fix a lattice with a fixed spacing $h$ which corresponds to a fixed UV-cutoff. We can then interpret $\mu_{\max}\rightarrow +\infty$ as the limit $\alpha\rightarrow +\infty$ at $\alpha_0$ fixed.
        \item The bounds derived in Proposition \ref{triviality_prop} could be sharpened. For the triviality statement, they are sufficient.
    \end{itemize}
\end{remarks}

\subsection{Uniqueness of the mean-field trivial solution}
\label{Uniqueness_trivial_solution}

\paragraph{} So far, we proved that for the mean-field $O(N)$-model $\varphi^4_4$-theory has a trivial solution for fixed mean-field condition. Here we prove the uniqueness of the trivial solution we constructed. We restrict for simplicity of notation to the case $N=1$. The more general case $N>1$ can be treated analogously.

The mean-field FE (\ref{MFE_O(N)}) can be obtained again following Felder's steps \cite{Felder} by considering the continuum limit of the hierarchical model, introduced by Dyson \cite{Dyson_hierarchical_model}. The effective action at the scale $L^{-1}\lambda$, where $L>1$, is related to the effective action at the scale $\lambda$ by

\begin{equation}\label{Felders_functional}
    e^{-u(L^{-1}\lambda,x)}=\int d\mu_L(y) e^{-L^{4} u(\lambda, L^{-1}x+y)}\;,\quad \lambda\in (0,\Lambda_0]\;,
\end{equation}
where $\mu_L$ is the one-dimensional Gaussian measure defined by

\begin{equation}
    d\mu_L(y) := \dfrac{1}{\sqrt{2\pi (L-1)}}e^{-\frac{y^2}{2(L-1)}}dy\;.
\end{equation}

Both the r.h.s and the l.h.s of (\ref{Felders_functional}) have a limit when $L\longrightarrow 1$, since the Gaussian measure $\mu_L$ becomes a Dirac measure in the limit $L\longrightarrow 1$. Then taking the $L$-derivative of (\ref{Felders_functional}) and evaluating at $L=1$ yields the partial differential equation
\begin{equation}\label{Felder_PDE}
    -\lambda\partial_\lambda u=\frac{1}{2}\partial_{xx}u-\frac{1}{2}\left(\partial_x u\right)^2+4u-x\partial_x u\;.
\end{equation}

If we expand $u(\lambda,x)$ as a power series in $x$
\begin{equation}\label{Felder_momentum_expansion}
    u(\lambda,x)=\sum_{n\in 2\N}\dfrac{(2)^{\frac{n}{2}} f_n(\lambda)}{n}x^n\;,
\end{equation}
then the moments $f_n(\lambda)$ satisfy the dynamical system
\begin{equation}
    -\frac{1}{n(n+1)}\lambda\partial_\lambda f_n =f_{n+2}-\frac{1}{n+1}\sum_{n_1+n_2=n+2}f_{n_1} f_{n_2} +\frac{4}{n(n+1)}f_n-\frac{1}{n+1}f_n\;.
\end{equation}
Setting $\lambda=\Lambda_0 e^{-\frac{\mu}{2}}$, we obtain again the FE (\ref{MFE_O(N)}). In \cite{Felder}, Felder derived (\ref{Felder_PDE}) in two ways:
\begin{itemize}
    \item Simplyfing Wilson's renormalization group equations \cite{Wilson1},\cite{Wilson2} using the local potential approximation \cite{ZUMBACH1994754}.
    \item Considering the continuum limit of the recursion relation of the hierarchical model introduced by Gallavotti \cite{Gavallotti_hierarchical_model}.
\end{itemize}   

He analyzed rigourously the global solutions of (\ref{Felder_PDE}) in great generality and concluded that in $d=4-\varepsilon$, the non-trivial fixed point solution $u_4$ in $3<d<4$ dimensions vanishes. Nevertheless his analysis does not exclude the existence of fixed points other than those he found. The momentum expansion (\ref{Felder_momentum_expansion}) may not be valid for arbitrarily large $x\in\R$. We prove that for the mean-field moments $f_n(\mu)$ constructed in Sect.\ref{Existence_trivial_solution}, $u(\lambda,x)$ is locally analytic w.r.t. $x$. We proceed as follows: first we bound $\partial_\mu^lf_2(\mu)$, then we bound inductively $\partial_\mu^l f_n(\mu)$ using the mean-field FE (\ref{MFE_O(N)}). Finally we obtain bounds for $f_n(\mu)$. We introduce the new variable $X:=n\mu$ and we define

\begin{equation}\label{definition_p_n}
    p_n(X)=\dfrac{X^{n-1}}{1+X^n}\;.
\end{equation}

\begin{proposition}\label{derivatives_p_n}
    We have for $l\in\N_0$ and $n\in\N$
    \begin{equation}
        \vert\partial_X^l p_n(X)\vert\leq \left\{
        \begin{array}{ll}
         \dfrac{l! \;  3^{l+1}e^{3l}}{n^{l+1}\mu^{2l+1}} & X\in (0,3) \\
        \dfrac{3^{l+1} l!}{\mu^{l+1} n^{l+1}} & X\geq 3\;.
        \end{array}
        \right.
    \end{equation}
\end{proposition}

\begin{proof}
    For $X\in (0,3)$, the proof is done by induction in $l\in\N_0$. The case $l=0$ is obvious. For $l> 0$, we use (\ref{derivative_of_f/g_formula}). Inserting the induction hypothesis in the r.h.s of (\ref{derivative_of_f/g_formula}) gives
    \begin{equation}
        \begin{split}
            \vert\partial_X^l p_n(\mu)\vert &\leq \dfrac{1}{1+X^n}\Bigg[\prod_{i=0}^{l-1}(n-1-i)X^{n-1-l}+l!\sum_{j=1}^{l}\dfrac{\prod_{i=0}^{l-j}(n-i)X^{n-l-1+j}}{(l+1-j)!\; (j-1)!}  \frac{(j-1)!\;3^{j}e^{3(j-1)}}{n^{j}\mu^{2j-1}}\Bigg] \\
            &\leq \dfrac{l!}{n^{l+1}}\Bigg[\dfrac{3^l}{\mu^{2l+1} l!}+3^{l+1}\sum_{j=1}^l \dfrac{e^{3(j-1)}}{(l+1-j)!} \frac{1}{\mu^{2l+1}}\Bigg] \\
            &\leq \dfrac{l! \;  3^{l+1}e^{3l}}{n^{l+1}\mu^{2l+1}}.
        \end{split}
    \end{equation}

For $X\geq 3$, we expand $p_n(X)$ as a power series in $\frac{1}{X}$, then we have
\begin{equation}
\begin{split}
  \vert\partial_X^l p_n(X)\vert &\leq \sum_{k=0}^\infty \dfrac{(nk+l)!}{(nk)!}\dfrac{1}{X^{nk+l+1}}  \\
  &\leq \dfrac{2^l l!}{X^{l+1}}\sum_{k=0}^\infty \dfrac{2^{nk}}{X^{nk}}\leq \dfrac{3^{l+1} l!}{\mu^{l+1} n^{l+1}}\;.
\end{split}
\end{equation}
\end{proof}

Now we prove bounds for the derivatives of $f_2(\mu)$

\begin{proposition}\label{more_precise_bounds_derivatives_f_2}
We have for a constant $K_1(1,K)$
    \begin{equation}\label{bounds_generic_derivatives_f_2}
    \left\vert\partial_\mu^l f_2(\mu)\right\vert\leq  \frac{K_1(1,K)^{l+1} l!\;}{M_l(\mu)}\;,\quad l\geq 0,\quad \mu\in (0,\mu_{\max}]\;,
\end{equation}
where we defined
\begin{equation}
    M_l(\mu):=\min\lbrace \mu^{2l+1},\mu^l\rbrace\;.
\end{equation}
\end{proposition}

\begin{proof}
From Proposition \ref{derivatives_p_n} and Proposition \ref{b_n_bounds}

\begin{equation}\label{derivatives_f_2_more_precise}
\begin{split}
 \vert \partial_\mu^l f_2(\mu)\vert &\leq C(1,K)\Bigg[\sum_{n<\frac{3}{\mu}}\dfrac{n}{2^n}\dfrac{l!\;3^{l+1}e^{3l}}{\mu^{2l+1}}+\sum_{n\geq\frac{3}{\mu}}\dfrac{n^2}{2^n}\dfrac{l!\; 3^{l+1}}{\mu^l}\Bigg] \\
 &\leq \frac{K_1(1,K)^{l+1} l!\;}{M_l(\mu)}\;,
\end{split}
\end{equation}
if we choose $K_1(1,K)>24 C(1,K)$.

\end{proof}

\paragraph{} Now we prove bounds for the derivatives of $f_n(\mu)$. It is convenient to distinguish $\mu<1$ and $\mu\geq 1$. We have

\begin{proposition}\label{bounds_f_n_more_precise}
    Let $f_n(\mu)$ be smooth mean-field solutions of the FE (\ref{MFE_O(N)}) and we assume that the derivatives of the two-point function $\partial_\mu^l f_2(\mu)$ satisfy the bounds (\ref{bounds_generic_derivatives_f_2}). Then we have for a constant $K_2(1,K)>K_1(1,K)$
    \begin{equation}
        \vert\partial_\mu^l f_n(\mu)\vert\leq \dfrac{K_2(1,K)^{n+l-1}}{(l+1)^2}\dfrac{(n+l)!\;}{n!}\dfrac{1}{\mu^{2l+n-1}},\quad n\geq 2,\;, l\geq 0\;, \mu<1\;,
    \end{equation}
    and
    \begin{equation}
        \vert\partial_\mu^l f_n(\mu)\vert\leq \dfrac{K_2(1,K)^{n+l-1}}{(l+1)^2}\dfrac{(n+l)!\;}{n!},\quad n\geq 2,\;, l\geq 0\;, \mu\geq 1\;.
    \end{equation}
\end{proposition}
\begin{proof}
    The proof is done by induction in $n+l$, going up in $n$. We will prove our statement for $0<\mu<1$. For $n=2$, bounds follow from (\ref{bounds_generic_derivatives_f_2}) since $M_l(\mu)=\mu^{2l+1}$. For $n>2$ we differentiate (\ref{MFE_O(N)}) $l$ times w.r.t. $\mu$ and we insert the induction hypothesis. We get
    \begin{equation}
        \begin{split}
            \vert\partial_\mu^l f_{n+2}(\mu)\vert &\leq \frac{2}{n(n+1)}\dfrac{K_2(1,K)^{n+l}}{(l+2)^2}\dfrac{(n+l+1)!\;}{n!}\frac{1}{\mu^{2l+2+n-1}} \\
            &+ \frac{1}{n+1} \dfrac{K_2(1,K)^{n+l-1}}{(l+1)^2}\dfrac{(n+l)!\;}{n!}\frac{1}{\mu^{2l+n-1}} \\
            &+ \dfrac{K_2(1,K)^{n+l}}{n+1}\sum_{n_1+n_2=n+2}\sum_{l_1+l_2=l}\binom{l}{l_1}\dfrac{(n_1+l_1)!\; (n_2+l_2)!}{(l_1+1)^2(l_2+1)^2}\dfrac{1}{n_1!\; n_2!}\dfrac{1}{\mu^{2l+n}}\;.
        \end{split}
    \end{equation}
We will proceed term by term.
\begin{itemize}
    \item First term: it is bounded by
    \begin{equation}\label{prop_3_8_first}
        \dfrac{K_2(1,K)^{n+l+1}}{(l+1)^2}\dfrac{(n+l+2)!\;}{(n+2)!}\dfrac{1}{\mu^{2l+n+1}}\dfrac{2}{K_2(1,K)}\;.
    \end{equation}
    \item Second term: since $\mu<1$, it is bounded by
    \begin{equation}\label{prop_3_8_second}
        \dfrac{K_2(1,K)^{n+l+1}}{(l+1)^2}\dfrac{(n+l+2)!\;}{(n+2)!}\dfrac{1}{\mu^{2l+n+1}}\dfrac{1}{4 K_2(1,K)^2}\;.
    \end{equation}
    \item Third term: We use the Vandermonde inequality 
     \begin{equation}\label{Vandermonde}
            \binom{l}{l_1}\binom{n+2}{n_1}\leq\binom{n+l+2}{n_1+l_1}\;.
        \end{equation}
        Then  we obtain
        \begin{equation}
        \begin{split}
          \binom{l}{l_1}\dfrac{(n_1+l_1)!\;(n_2+l_2)!}{n_1!\;n_2!} &=\binom{l}{l_1}\binom{n+2}{n_1}\dfrac{(n_1+l_1)!\;(n_2+l_2)!}{(n+2)!}\\  
          &\leq \dfrac{(n+l+2)!}{(n+2)!}\;.
        \end{split}
        \end{equation}

        Since 
    \begin{equation}
        \sum_{l_1+l_2=l}\dfrac{1}{(l_1+1)^2 (l_2+1)^2}\leq\dfrac{1}{(l+2)^2}\left[2\sum_{l_1=0}^l\frac{1}{(l_1+1)^2}+\dfrac{2}{l+2}\sum_{l_1=0}^l\dfrac{1}{l_1+1}\right]\;,
    \end{equation}
       then we obtain the following bound
        \begin{equation}\label{prop_3_8_third}
        \dfrac{K_2(1,K)^{n+l+1}}{(l+1)^2}\dfrac{(n+l+2)!\;}{(n+2)!}\dfrac{1}{\mu^{2l+n+1}}\dfrac{6}{K_2(1,K)}\;,
    \end{equation}
    using again the fact that $0<\mu<1$.
\end{itemize}
Summing the three bounds (\ref{prop_3_8_first}), (\ref{prop_3_8_second}) and (\ref{prop_3_8_third}) yields
\begin{equation}
\begin{split}
   \vert \partial_\mu^l f_n(\mu)\vert &\leq \left[\dfrac{8}{K_2(1,K)}+\dfrac{1}{4 K_2(1,K)^2}\right] \dfrac{K_2(1,K)^{n+l+1}}{(l+1)^2}\dfrac{(n+l+2)!\;}{(n+2)!}\dfrac{1}{\mu^{2l+n+1}} \\
   &\leq\dfrac{K_2(1,K)^{n+l+1}}{(l+1)^2}\dfrac{(n+l+2)!\;}{(n+2)!}\dfrac{1}{\mu^{2l+n+1}}\;.
\end{split}
    \end{equation}
For $\mu\geq 1$, we can bound the negative power of $\mu$ in (\ref{bounds_generic_derivatives_f_2}) by $1$ and then we proceed as above.
\end{proof}

\paragraph{} From Proposition \ref{bounds_f_n_more_precise}, the mean-field effective action $u(\mu,x)$ is locally analytic w.r.t. $x$ for $\mu>0$. Its radius of convergence $R(\mu)$ is such that 
\begin{equation}
R(\mu)\geq \dfrac{\min\lbrace \mu,1\rbrace}{\sqrt{2}K_2(1,K)},\quad \mu\in (0,\mu_{\max}]\;.
\end{equation}
For $\mu=0$, the boundary conditions (\ref{BDY_trivial_field}) imply that $u(0,x)$ is polynomial in $x$, the solution $u(\mu,x)$ is well-defined for $\mu\in [0,\mu_{\max}]$. Then we claim that for fixed mean-field boundary conditions, the solution $u(\mu,x)$ is unique.

\begin{theorem}[Uniqueness of the mean-field trivial solution for the pure $\varphi_4^4$-theory]\label{uniqueness_N=1}
We consider the bare interaction lagrangian of a pure $\varphi^4_4$-theory (\ref{bare_lagrangian_trivial}). For fixed mean-field boundary conditions (\ref{BDY_trivial_field}), consider smooth mean-field solutions $f_n(\mu)$ of the mean-field FE (\ref{MFE_O(N)}) which satisfy (\ref{BDY_trivial_field}) such that the corresponding mean-field effective action $u(\lambda,x)$ is locally analytic w.r.t. $x\in\R$. Then they are unique.
\end{theorem}
\begin{proof}
    For fixed mean-field boundary conditions (\ref{BDY_trivial_field}), let $f_n(\mu)$ be the mean-field solutions of (\ref{MFE_O(N)})  constructed in Sect.\ref{Existence_trivial_solution} which satisfy (\ref{BDY_trivial_field}). Let $\Tilde{f}_n(\mu)$ be solutions of the mean-field FE (\ref{MFE_O(N)}) which satisfy $f_n(0)=\Tilde{f}_n(0)$. We assume that the corresponding mean-field effective action $\Tilde{u}(\lambda,x)$ is locally analytic w.r.t. $x$ for $\lambda<\Lambda_0$. Then $u(\Lambda_0,z)=\Tilde{u}(\Lambda_0,z)$ for any arbitrary $z\in\R$ since they are polynomial in $z$. It follows from (\ref{Felders_functional}) that for a fixed $\lambda<\Lambda_0$, $u(\lambda,x)=\Tilde{u}(\lambda,x)$ for $0<\vert x\vert\leq\varepsilon_\lambda$, $\varepsilon_\lambda>0$. Thus it implies that $f_n(\mu)=\Tilde{f}_n(\mu)$.
\end{proof}

\paragraph{} The extension to the more general case $N>1$ is done by considering the mean-field effective action $u(\lambda,\bm{x})$, $\bm{x}\in\R^N$. We recall the Euclidean scalar product in $\R^N$

\begin{equation}\label{introduction_scalar_product}
    (\bm{x},\bm{y}):=\sum_{i=1}^N x_i y_i,\quad \vert\bm{y}\vert ^2:=(\bm{y},\bm{y})=\sum_{i=1}^N y_i^2 \;.
\end{equation}

The relation (\ref{Felders_functional}) is generalized as follows
\begin{equation}
    e^{-u(L^{-1}\lambda,\bm{x})}=\int d\mu_{N,L}(\bm{y}) e^{-L^{4} u(\lambda, L^{-1}\bm{x}+\bm{y})}\;,\quad \lambda\in (0,\Lambda_0]\;,
\end{equation}
where $\mu_{N,L}$ is the $N$-dimensional Gaussian measure defined by

\begin{equation}
    d\mu_{N,L}(\bm{y}) := \dfrac{1}{\sqrt{2\pi (L-1)}}e^{-\frac{\vert \bm{y}\vert ^2}{2(L-1)}}\prod_{i=1}^N dy_i\;.
\end{equation}

The generalization of (\ref{Felder_PDE}) is 
\begin{equation}\label{Felder_dim_N_PDE}
    -\lambda\partial_\lambda u=\frac{1}{2}\Delta u-\frac{1}{2}\vert\nabla u\vert^2+4u-(\bm{x},\nabla u)\;.
\end{equation}

Expanding $u(\lambda,\bm{x})$ as a power series in $\vert\bm{x}\vert$

\begin{equation}\label{power_series_Felder_dim_N}
    u(\lambda,\bm{x})=\sum_{n\in 2\N}\dfrac{(2)^{\frac{n}{2}} f_n(\lambda)}{n}\vert\bm{x}\vert^n\;,
\end{equation}
we find (\ref{MFE_O(N)}) upon setting $\lambda=\Lambda_0e^{-\frac{\mu}{2}}$. Extensions of Propositions \ref{derivatives_p_n}-\ref{bounds_f_n_more_precise} and Theorem \ref{uniqueness_N=1} to $N>1$ are immediate.

\subsection{The $1/N$-expansion}\label{Large_N_limit}

In cases where  $N$ may be considered to be large, the large $N$-expansion
in powers of $\frac{1}{N}$ is complementary to the perturbative expansion 
in the coupling constant. This expansion is based on rescaling the coupling 
constant as $g \to g/N\,$. 
Using this expansion the universal properties of critical systems obtained in 
an expansion in $\varepsilon=4-d$ can be obtained at fixed dimension 
but in an expansion w.r.t. $\frac{1}{N}$ instead \cite{Moshe2003}. 
Here we want to show as a cross-check that we can recover the behaviour 
 in $\frac{1}{N}$ in our framework. 
We choose the following bare interaction lagrangian 
\begin{equation}\label{bare_lagrangian_trivial_scaling}
    L_0(\varphi)=\int d^4x 
\Big( c_{0,2}\varphi^2(x)+\dfrac{c_{0,4}}{N}\varphi^4(x)\Big)\ .
\end{equation}
\begin{lemma}\label{triviality_lemma_1_N_large}
    Let $f_n$ be solutions of (\ref{MFE_O(N)}) which respect the 
boundary conditions (\ref{BDY_trivial_field}). We assume that for 
some $K$ sufficiently large:
\begin{equation}\label{initial_bounds_O(N)}
    \vert f_{2,0}\vert\leq\dfrac{\sqrt{K}}{4},
\quad \vert f_{4,0}\vert=\vert g_{4,0}\vert\leq\dfrac{\sqrt{K}}{32N}\;.
\end{equation}
Then
\begin{equation}
    \vert f_{2,1}\vert\leq \dfrac{K}{2},
\quad\vert g_{4,1}\vert\leq \dfrac{K}{32N}\;,
\end{equation}
\begin{equation}
    \vert g_{n,0}\vert\leq\dfrac{K^{\frac{n}{2}-\frac{3}{2}}}{2N^{\frac{n}{2}-1}n^2},
\quad \vert g_{n,1}\vert\leq \dfrac{K^{\frac{n}{2}-\frac{3}{2}}}{N^{\frac{n}{2}-1}n^2}
\left(1+\dfrac{nK}{2}\right)\; , \quad n \ge 6\; ,
\end{equation}
Furthermore
\begin{equation}\label{1Ngnk}
    \vert g_{n,k}\vert\leq 
\dfrac{1}{N^{\frac{n}{2}-1}}K^{\frac{n}{2}+k-\frac{3}{2}}
\left\vert\dfrac{n}{4}+k-3\right\vert!\;\dfrac{1}{(k!)^{\frac{1}{4}}}\;,
\quad \vert f_{2,k}\vert\leq 
K^{k+\frac{1}{2}}\dfrac{\vert k-3\vert!}{(\vert k-1\vert!)^{\frac{1}{4}}}\;,
\quad n\geq 4,k\geq 0\ .
\end{equation}
\end{lemma}

\begin{proof}
We have
\begin{equation}
    \vert f_{2,1}\vert\leq \dfrac{N+2}{N}\dfrac{\sqrt{K}}{32}
+\dfrac{\sqrt{K}}{4}+\dfrac{K}{16}\leq\frac{K}{2}\;,
\quad \vert g_{4,1}\vert\leq 4\dfrac{\sqrt{K}}{4}
\dfrac{\sqrt{K}}{32N}=\dfrac{K}{32N}\;,
\end{equation}
since $N\geq 1$. We can then proceed by induction in $n\,$, 
and the r.h.s of (\ref{FE_trivial_field2_1_O(N)}),(\ref{FE_trivial_field2_2_O(N)}) give the correct bound w.r.t. $N\,$
on $g_{n,0}\,,\ g_{n,1}\,$. 
The proof of (\ref{1Ngnk})
is identical to the proof of Proposition \ref{triviality_prop} 
up to the following changes: in the r.h.s. of (\ref{FE_trivial_field1O(N)}), 
the factor $(n+N)g_{n+2,k}$ produces a term $\frac{n+N}{N}$ 
which is obviously bounded by $n+1$. 
For the bound on $f_{2,k+1}$, the term $(N+2)g_{4,k}$ is bounded by  
$\frac{N+2}{N}\leq 3$. It is then easy to check that the claimed behaviour 
w.r.t. $N\,$ is true.
\end{proof}
Due to these bounds, the bounds (\ref{bound_b_n}) still hold and 
using the results in Sect.\ref{Triviality} we construct the trivial 
solution as before. The behavior in $N$ of the bounds derived in Proposition \ref{triviality_lemma_1_N_large} is sharp, and we see that the two-point function $f_2$ does not blow up w.r.t. $N$ while the four-point function behaves as $\frac{1}{N}$ in the large $N$ limit, in agreement with the results obtained from partial resummed perturbation theory \cite{Amit}.

\section{The case of the theory with a physical IR cutoff}\label{Massive_theory}

\paragraph{} In the previous section we constructed the trivial solution using a technical  IR cutoff supposed to take the role of the mass
and set the mass equal to zero. Here we will work with the true propagator of the massive theory. We choose a regularized flowing propagator which preserves the analyticity properties w.r.t. $\alpha\,$.
We will again prove triviality of the mean-field massive $\varphi^4_4$-theory. 
Here we restrict to the case $N=1\,$ in order not to overload the proof with
technicalities. But it also goes through for $N>1\,$. 

We will adapt the flow equations to the new scheme. The trivial solution  will be constructed
using again the ansatz introduced in (\ref{ansatz}). We follow
 the steps as from Sect.\ref{Existence_trivial_solution}.

\subsection{The flow equations for the massive theory}
\label{Massive_flow_equations_general}

We assume $\alpha_0<\frac{1}{2m^2}$. 
We choose the following regularized propagator
\begin{equation}\label{propagator}
    \Tilde{C}^{\alpha_0,\alpha}(p,m)
=\dfrac{e^{-\alpha_0(p^2+m^2)}-e^{-\alpha(p^2+m^2)}(\frac{1}{m^2}+\alpha_0-\alpha)m^2}
{p^2+m^2}\;,\quad\alpha\in\left[\alpha_0,\frac{1}{m^2}+\alpha_0\right]\ .
\end{equation}
We also verify the required properties
\begin{equation}
 \Tilde{C}^{\alpha_0,\alpha_0}(p,m)=0\ ,\quad   
\lim\limits_{\alpha_0\rightarrow 0}\lim\limits_{\alpha\rightarrow \frac{1}{m^2}+\alpha_0}
  \Tilde{C}^{\alpha_0,\alpha}(p,m)=\frac{1}{p^2+m^2}\ ,\quad
 \Tilde{C}^{\alpha_0,\alpha}(p,m) \ge 0 \ .
\end{equation}
At fixed $\alpha$, $\Tilde{C}^{\alpha_0,\alpha}(p,m)$ falls off more rapidly 
than any power of $\vert p\vert$ and is smooth and locally 
analytic w.r.t. $\alpha$. We set $\beta_0:=\alpha_0 m^2$ and $\beta:=\alpha m^2$. Then we have
\begin{equation}
    \Tilde{C}^{\beta_0,\beta}(p,m)=\dfrac{e^{-\frac{\beta_0}{m^2}(p^2+m^2)}-e^{-\frac{\beta}{m^2}(p^2+m^2)}(1+\beta_0-\beta)}
{p^2+m^2}\;,\quad\beta\in\left[\beta_0,1+\beta_0\right]\ .
\end{equation}

Proceeding as before, see Sect.\ref{Flow_equations_O(N)_model_general}, 
we obtain the flow equations for the CAS in expanded form as
\begin{equation}\label{FE-massif2}
        \begin{split}
    \partial_\beta \mathcal{L}^{\beta_0,\beta}_n(p_1,\cdots,p_n) 
&= \binom{n+2}{2}\int_k \Dot{\tilde{C}}^\beta (k,m)
\mathcal{L}^{\beta_0,\beta}_{n+2}(k,-k,p_1,\cdots,p_n) \\
    &-\dfrac{1}{2}\sum_{n_1+n_2=n+2}n_1 n_2
\mathbb{S}\Bigg(\mathcal{L}^{\beta_0,\beta}_{n_1}(p_1,\cdots,p_{n_1-1},q)
\Dot{\tilde{C}}^\beta (q,m)
\mathcal{L}^{\beta_0,\beta}_{n_2}(-q,p_{n_1},\cdots,p_{n})\Bigg)\;,
\end{split}
\end{equation}
where $\Dot{\tilde{C}}^\beta (k,m) = \partial_{\beta}
\Tilde{C}^{\beta_0,\beta}(k,m)\,$.
In the mean field approximation, we substitute 
$\mathcal{L}^{\beta_0,\beta}_n(p_1,\cdots,p_n)$ with 
$A_{n}^{\beta_0,\beta}:=\mathcal{L}^{\beta_0,\beta}_n(0,\cdots,0)$. We obtain
from (\ref{FE-massif2})
\begin{equation}\label{FE-MF1}
\begin{split}
 \partial_\beta A_{n}^{\beta_0,\beta}
&=\binom{n+2}{2}I(\beta)A_{n+2}^{\beta_0,\beta}
 \\
 &-\frac{1}{2m^2}e^{-\beta}(2+\beta_0-\beta)
\sum_{n_1+n_2=n+2}n_1 n_2 A_{n_1}^{\beta_0,\beta}A_{n_2}^{\beta_0,\beta}\;,
\quad\beta\in[\beta_0,1+\beta_0]\;,   
\end{split}
\end{equation}
where $I(\beta):=m^2c(1+\beta_0-\beta)
\frac{e^{-\beta}}{\beta^2}+m^2\int_k\frac{e^{-\beta(k^2+1)}}{k^2+1}$. 

\begin{remark}
For the regularized propagator (\ref{Propagator_reg_1}) 
we would have to substitute $I(\beta)$ 
with $\,c\,\frac{e^{-\beta}}{\beta^2}\,$.\\
Therefore bounding $A^{\beta_0,\beta}_{n+2}$ from $A^{\beta_0,\beta}_{n'}$, $n'\leq n$ 
and from $\partial_\beta A^{\beta_0,\beta}_{n}$ would produce bounds, on dividing by $I(\beta)$, which blow up for $\beta\longrightarrow +\infty$. 
For the choice (\ref{propagator})  $\beta$ is limited by $ 1+\beta_0$.
\end{remark}

We again factor out the scaling factor $\beta^{\frac{n}{2}-2}$ 
and the combinatorial factor  setting
\begin{equation}
    A_{n}^{\beta_0,\beta}
:=\beta^{\frac{n}{2}-2}e^{\frac{\beta n}{2}}\frac{1}{m^{n-4}}\frac{1}{n}
a_{n}(\beta)\;,\quad\beta\in [\beta_0,1+\beta_0]\;,
\end{equation}
where we removed the upper index $\beta_0$ on the r.h.s for shortness. 
Then the mean-field system (\ref{FE-MF1}) reads
\begin{equation}\label{FE-MF2}
\begin{split}
   G(\beta)\ a_{n+2}(\beta) 
&= \dfrac{2}{n(n+1)}\, \beta\partial_\beta 
a_{n}(\beta)+\dfrac{n-4}{n(n+1)}a_{n}(\beta)
+\frac{1}{n+1}\,\beta\, a_{n}(\beta) \\
   &+\dfrac{1}{n+1}(2+\beta_0-\beta)
\sum_{n_1+n_2=n+2} a_{n_1}(\beta)\,a_{n_2}(\beta)\;,
\end{split}
\end{equation}
where $G(\beta):=c\, (1+\beta_0-\beta)\,+\,\beta^2 \ e^\beta 
\int_k\frac{e^{-\beta(k^2+1)}}{k^2+1}\ $. 
The integral appearing in the expression of $G(\beta)$  can be rewritten as 
\[
\beta^2 e^\beta \int_k\frac{e^{-\beta(k^2+1)}}{k^2+1}=\, 
2\,c\,\beta\int_0^{+\infty}
\dfrac{u^3 e^{-u^2}}{u^2+\beta}\;du\;,
\]
so that the limit $\beta\rightarrow 0$ is finite. 
Moreover it is easy to see that $G(\beta_0)=c+\mathcal{O}(\beta_0)$ 
when $\beta_0\rightarrow 0$, 
meaning that $G(\beta)$ takes a role analogous to $c$ in the theory at $m=0\,$, 
when we compare (\ref{FE-MF2}) to (\ref{FE_O(N)_4}) for $N=1$.

We perform a change of variable defining 
$\mu:=\ln\left(\dfrac{\beta}{\beta_0}\right)$ 
so that $\beta\partial_\beta=\partial_\mu$. Setting 
$f_{n}(\mu)=a_{n}(\beta)\,$ we get

\begin{equation}\label{FE-mu}
\begin{split}
    H(\mu)f_{n+2}(\mu) &= \dfrac{2}{n(n+1)}
\partial_\mu f_n(\mu)+\dfrac{n-4}{n(n+1)}f_n(\mu)
+\frac{1}{n+1}\, \beta_0\ e^\mu\, f_n(\mu) \\
    &+\dfrac{1}{n+1}(2+\beta_0\,(1-e^\mu))
\sum_{n_1+n_2=n+2} f_{n_1}(\mu)f_{n_2}(\mu)\;,
\quad \mu\in\left[0,\tilde{\mu}_{\max}\right]\;,
\end{split}
\end{equation}
with $\,\tilde{\mu}_{\max}:=\ln\left(1+\dfrac{1}{\beta_0}\right)$ 
and 
\begin{equation}\label{H(mu)}
\,H(\mu):=G(\beta_0e^\mu)=c(1+\beta_0)-c\,\beta_0\,e^\mu+h(\mu)\;,    
\end{equation}

 where we set
\begin{equation}\label{Hh}
    h(\mu)=\beta_0^2\, e^{2\mu}\int_k\dfrac{e^{-\beta_0e^\mu k^2}}{k^2+1}\ .
\end{equation}

From Lemma \ref{boundinverseofH}, we can absorb the factor $H(\mu)$ 
in a new non-singular change of function
\begin{equation}\label{definition_tilde_f}
\tilde{f}_n(\mu)=(H(\mu))^{\frac{n}{2}-1} f_n(\mu)\ .    
\end{equation}

Then, the dynamical system (\ref{FE-mu}) reads
\begin{equation}\label{FE-mu_new}
\begin{split}
    \tilde{f}_{n+2}(\mu) 
&= \dfrac{2}{n(n+1)}\partial_\mu \tilde{f}_n(\mu)
+\dfrac{n-4}{n(n+1)}\tilde{f}_n(\mu)
-\dfrac{n-2}{n(n+1)}\partial_\mu\log(H(\mu))\tilde{f}_n(\mu) \\
&+\frac{1}{n+1}\beta_0 e^\mu \tilde{f}_n(\mu)+\dfrac{1}{n+1}(2+\beta_0(1-e^\mu))
\sum_{n_1+n_2=n+2} \tilde{f}_{n_1}(\mu)\tilde{f}_{n_2}(\mu)\;,
\end{split}
\end{equation}
which can also be written as 
\begin{equation}\label{tildef_4}
    \tilde{f}_4(\mu)=\dfrac{1}{3}
\Big(\partial_\mu\tilde{f}_2(\mu)-\Tilde{f}_2(\mu)
+\beta_0\ e^\mu\,\Tilde{f}_2(\mu)\,+\,(2+\beta_0\, (1-e^\mu))
\Tilde{f}_2^2(\mu)\Big)
\end{equation}
\begin{equation}\label{tildef_n,n>2}
    \begin{split}
         \tilde{f}_{n+2}(\mu) 
&= \dfrac{2}{n(n+1)}\partial_\mu \tilde{f}_n(\mu)
+\dfrac{n-4}{n(n+1)}
\tilde{f}_n(\mu)-\dfrac{n-2}{n(n+1)}
\partial_\mu\log(H(\mu))\tilde{f}_n(\mu) \\
&+\frac{1}{n+1}\,\beta_0\ e^\mu\, 
\tilde{f}_n(\mu) +\dfrac{2}{n+1}\,(2+\beta_0\,(1-e^\mu))\,
\Tilde{f}_2(\mu)\Tilde{f}_n(\mu) \\
&+\dfrac{1}{n+1}\,(2+\beta_0\,(1-e^\mu))
\!\!\! \sum_{\substack{n_1+n_2=n+2\\
n_i\geq 4}} \tilde{f}_{n_1}(\mu)\tilde{f}_{n_2}(\mu),\quad n\geq 4\ .
    \end{split}
\end{equation}
The flow equations (\ref{tildef_4})-(\ref{tildef_n,n>2}) 
include additional terms which are $\mu$-dependent as compared
to the flow equations (\ref{MFE_O(N)}), but they retain a similar form.

\subsection{Triviality of massive $\varphi_4^4$ theory}
\label{Triviality_massive_theory}

For the triviality proof we proceed in close analogy with
Sect.\ref{Existence_trivial_solution}. Most technical 
proofs are deferred to the Appendix.
 Due to Lemma \ref{derivatives_logH}, $\partial_\mu\log (H(\mu))$ 
is locally analytic around $\mu=0\,$. For $|\mu|$ sufficiently
small 
\begin{equation}\label{expansion_log(H)}
 \!\!\!
 \partial_\mu\log (H(\mu))=\sum_{k\geq 0} h_k \mu^k\;,
\ \ h_k=\dfrac{\partial_\mu^{k+1}\log(H(\mu))\vert_{\mu=0}}{k!}\ ,\ \  
\vert h_k\vert\leq c\, C^{k+1}\,(5e)^{k+2}2^{k+1}\  .
\end{equation}

The bare mean-field boundary conditions for 
$\tilde{f}_n(\mu)$ are
\begin{equation}\label{BDYconditions_massive}
     \tilde{f}_2(0)=2(2\pi)^4\alpha_0 e^{-\beta_0}c_{0,2},
\quad \tilde{f}_4(0)=(2\pi)^4 e^{-2\beta_0}H(0) c_{0,4},
\quad \tilde{f}_n(0)=0,\quad n\geq 6\ .
\end{equation}
First we will factor out a power of 
$\mu$ in $\Tilde{f}_n(\mu)$ for $n\geq 4$.

\begin{restatable}{lemma}{factorization}\label{factorization_tilde_f_n}
     For smooth solutions $\tilde{f}_n(\mu)$ of (\ref{tildef_n,n>2}) 
with boundary conditions (\ref{BDY_trivial_field}), we have
    \begin{equation}
        \partial_\mu^l \tilde{f}_n(0)=0\;,\quad n\geq 6,\; 
0\leq l\leq\frac{n}{2}-3\ .
    \end{equation}  
\end{restatable}

\begin{proof}
    See Appendix \ref{Appendix_D_lemmas_triviality}.
\end{proof}

From Lemma \ref{factorization_tilde_f_n}, we can write
\begin{equation}
    \tilde{g}_n(\mu)=\mu^{2-\frac{n}{2}}\Tilde{f}_n(\mu),\quad n\geq 4\;,
\end{equation}
where $\Tilde{g}_n(\mu)$ is smooth w.r.t. $\mu$. 
The FEs (\ref{tildef_n,n>2}) can be written in terms of $\Tilde{g}_n(\mu)$
\begin{equation}\label{FE-tilde-g_n(mu)}
   \begin{split}
    \mu^2 \tilde{g}_{n+2} &=\dfrac{n-4}{n(n+1)}
\tilde{g}_n+\dfrac{2}{n(n+1)}\mu\partial_\mu 
\tilde{g}_n+\dfrac{1}{n+1}(2+\beta_0-\beta_0e^\mu)\sum_{\substack{n_1+n_2=n+2\\
n_i\geq 4}} \tilde{g}_{n_1}\tilde{g}_{n_2}\\
                  &+\mu\dfrac{1}{n+1}
\tilde{g}_n\left(2(2+\beta_0-\beta_0 \, e^\mu)\tilde{f}_2
+\beta_0e^\mu+1-\dfrac{4}{n}-(1-\frac{2}{n})\partial_\mu \log(H(\mu))\right)\ .
\end{split}
\end{equation}
We write the formal Taylor expansion of $\Tilde{f}_2$ 
and $\Tilde{g}_n$ around $\mu=0$
\begin{equation}\label{expansion_tilde}
    \Tilde{f}_2(\mu)=\sum_{k\geq 0}
\Tilde{f}_{2,k}\mu^k\;,\quad \Tilde{g}_n(\mu)=\sum_{k\geq 0}\Tilde{g}_{n,k}\mu^k\ .
\end{equation}

From (\ref{expansion_tilde}), (\ref{expansion_log(H)}) and the 
FEs (\ref{tildef_4}), (\ref{tildef_n,n>2}), 
we deduce the relations between the coefficients of the formal 
Taylor expansion of $\Tilde{f}_2$ and $\Tilde{g}_n$
\begin{equation}\label{FE_tilde_f_2_k}
\begin{split}
 \Tilde{f}_{2,k+1} &= \frac{1}{k+1}\Big(3\Tilde{g}_{4,k}
+\Tilde{f}_{2,k}-(2+\beta_0)
\sum_{\nu=0}^k\Tilde{f}_{2,\nu}\Tilde{f}_{2,k-\nu}
-\beta_0\sum_{\nu=0}^k\frac{1}{(k-\nu)!}
\Tilde{f}_{2,\nu} \\
&+\beta_0\sum_{\nu=0}^k\frac{1}{\nu!}\sum_{\nu'=0}^{k-\nu}
\Tilde{f}_{2,\nu'}\Tilde{f}_{2,k-\nu-\nu'}\Big)\ .   
\end{split}
\end{equation}

\begin{equation}\label{FE_tilde_g_n_k}
\begin{split}
    \Tilde{g}_{n,k+2} &=\dfrac{n(n+1)}{n+2k}\Tilde{g}_{n+2,k}
-\dfrac{n-4}{n+2k}\Tilde{g}_{n,k+1}+\dfrac{n-2}{n+2k}
\sum_{\nu=0}^{k+1}\Tilde{g}_{n,\nu} h_{k+1-\nu} \\
&-\dfrac{\beta_0 n}{n+2k}
\sum_{\nu=0}^{k+1}\frac{1}{(k+1-\nu)!}\Tilde{g}_{n,\nu}
-\dfrac{(2+\beta_0)n}{n+2k}\sum_{\substack{n_1+n_2=n+2\\
n_i\geq 4}}\sum_{\nu=0}^{k+2}\Tilde{g}_{n_1,\nu}\Tilde{g}_{n_2,k+2-\nu}
\\ 
&+\dfrac{\beta_0 n}{n+2k}\sum_{\substack{n_1+n_2=n+2\\
n_i\geq 4}}\sum_{\nu=0}^{k+2}\frac{1}{\nu!}
\sum_{\nu'=0}^{k+2-\nu-\nu'}\Tilde{g}_{n_1,\nu'}\Tilde{g}_{n_2,k+2-\nu-\nu'} \\
&- 2\dfrac{(2+\beta_0)n}{n+2k}\sum_{\nu=0}^{k+1}
\Tilde{g}_{n,\nu}\Tilde{f}_{2,k+1-\nu}
+\dfrac{2\beta_0 n}{n+2k}
\sum_{\nu=0}^{k+1}\frac{1}{\nu!}
\sum_{\nu'=0}^{k+1-\nu}\Tilde{g}_{n,\nu'}\Tilde{f}_{2,k+1-\nu-\nu'}\ .
\end{split}
\end{equation}

Regularity at $\mu=0$ implies that
\begin{equation}\label{tilde_g_n,0}
    0=\dfrac{n-4}{n}\Tilde{g}_{n,0}+2\sum_{\substack{n_1+n_2=n+2\\
n_i\geq 4}}\Tilde{g}_{n_1,0}\Tilde{g}_{n_2,0}\;.
\end{equation}
and
\begin{equation}\label{tilde_g_n_1}
\begin{split}
0 &=\dfrac{n-2}{n}\Tilde{g}_{n,1}+\dfrac{n-4}{n}\Tilde{g}_{n,0}
+\beta_0\Tilde{g}_{n,0}-\dfrac{n-2}{n}h_0\Tilde{g}_{n,0}+4\sum_{\substack{n_1+n_2=n+2\\
n_i\geq 4}}\Tilde{g}_{n_1,0}\Tilde{g}_{n_2,1} \\
&-\beta_0\sum_{\substack{n_1+n_2=n+2\\
n_i\geq 4}}\Tilde{g}_{n_1,0}\Tilde{g}_{n_2,0}+4\Tilde{g}_{n,0}\Tilde{f}_{2,0}\ .    
\end{split}
\end{equation}

Using (\ref{tilde_g_n,0}) we can rewrite (\ref{tilde_g_n_1}) as
\begin{equation}\label{true_tilde_g_n_1}
   \dfrac{n-2}{n}\Tilde{g}_{n,1}+4\sum_{\substack{n_1+n_2=n+2\\
n_i\geq 4}}\Tilde{g}_{n_1,0}\Tilde{g}_{n_2,1}+\Tilde{g}_{n,0}
\left(4\Tilde{f}_{2,0}+(1-\frac{4}{n})
(1+\frac{\beta_0}{2})+\beta_0-(1-\frac{2}{n})h_0\right)=0\ . 
\end{equation}

Now we derive bounds on $\Tilde{g}_{n,k}$ and $\Tilde{f}_{2,k}$.

\begin{restatable}{lemma}{boundslemma}\label{tilde_lemma_bounds}
    Let $\Tilde{f}_n$ be smooth solutions of 
(\ref{tildef_4}),(\ref{tildef_n,n>2}). 
For given $\tilde{f}_{2,0},\tilde{f}_{4,0}$ we choose $K$ 
sufficiently large such that
\begin{equation}\label{initial_bounds_tilde}
    \vert \tilde{f}_{2,0}\vert\leq\dfrac{\sqrt{K}}{16},
\quad \vert \tilde{f}_{4,0}\vert
=\vert \tilde{g}_{4,0}\vert\leq\dfrac{\sqrt{K}}{32}\ .
\end{equation}
Then
\begin{equation}
    \vert \tilde{f}_{2,1}\vert\leq \dfrac{K}{2},
\quad\vert \tilde{g}_{4,1}\vert\leq \dfrac{K}{32}\;,
\end{equation}
and for $n\geq 4$
\begin{equation}\label{general_bounds_n_tilde}
    \vert \tilde{g}_{n,0}\vert\leq\dfrac{K^{\frac{n}{2}-\frac{3}{2}}}{2n^2},
\quad \vert \tilde{g}_{n,1}\vert\leq \dfrac{K^{\frac{n}{2}-\frac{1}{2}}}{n}\ .
\end{equation}
\end{restatable}

\begin{proof}
    See Appendix \ref{Appendix_D_lemmas_triviality}.
\end{proof}

\begin{restatable}{proposition}{trivprop}\label{triviality_prop_massive}
    Under the same assumptions as in 
Lemma \ref{tilde_lemma_bounds}, choosing $K$ large enough we have
\begin{equation}
    \vert \tilde{g}_{n,k}\vert
\leq K^{\frac{n}{2}+k-\frac{3}{2}}
\left\vert\dfrac{n}{4}+k-3\right\vert!\;\dfrac{1}{(k!)^{\frac{1}{8}}}\;,
\quad \vert \tilde{f}_{2,k}\vert
\leq K^{k+\frac{1}{2}}
\dfrac{\vert k-3\vert!}{(\vert k-1\vert!)^{\frac{1}{8}}}\;,
\quad n\geq 4,\ k\geq 0\ .
\end{equation}
\end{restatable}

\begin{proof}
    See Appendix \ref{Appendix_D_lemmas_triviality}.
\end{proof}

\paragraph{} Choosing a smooth two-point function of the 
form (\ref{ansatz}), the sequence $(b_n)_{n\geq 1}$ satisfies
bounds of the same type as (\ref{bound_b_n}) ($N=1$ in our setting). Since we chose the same two-point function as in 
Sect.\ref{Triviality}, Lemma \ref{derivativesf_2} remains valid 
and so does Proposition \ref{vanishing_solutions_in_the_UV}.
Then, the solutions $\tilde{f}(\mu)$ are well-defined 
on $[0,\Tilde{\mu}^{\max}]$ and vanish in the UV-limit. 
Therefore the extension of Theorem \ref{theorem_triviality_massless} 
to the massive theory is straightforward.

\begin{theorem}[Triviality of pure mean-field $\varphi^4$-theory 
for the theory with a physical IR cutoff]\label{theorem_triviality_massive}
    Consider the $\varphi^4_4$-theory of bare interaction lagrangian (\ref{bare_lagrangian_trivial}) for $N=1$. Let $\tilde{f}_n(\mu)$ be  smooth solutions of the mean-field flow equations (\ref{FE-mu_new}) and the corresponding mean-field boundary conditions (\ref{BDYconditions_massive}) with
    \begin{equation}
        0<c_{0,4}<+\infty,\quad \vert c_{0,2}\vert<+\infty\;.
    \end{equation}
    There exist smooth solutions of (\ref{FE-mu_new}) $\tilde{f}_n(\mu)\in C^\infty([0,\tilde{\mu}_{\max}])$ such that they vanish in the UV-limit, i.e.
    \begin{equation}
        \lim\limits_{\tilde{\mu}_{\max}\rightarrow +\infty} \tilde{f}_n(\tilde{\mu}_{\max})=0,\quad n\geq 2\;.
    \end{equation} 
\end{theorem}

\begin{proof}
    The proof is the same as for Theorem \ref{theorem_triviality_massless}.
\end{proof}

The uniqueness of the solutions can be proven following the reasoning in Sect.\ref{Uniqueness_trivial_solution}. The differences remain purely technical as the coefficients in the r.h.s of (\ref{FE-mu_new}) are $\mu$-dependent analytic functions in our new setting.

\newpage

\appendix

\begin{appendices}

\section{Appendix A}
\subsection{Properties of Gaussian measures}\label{appendix_A}
We consider a Gaussian probability measure $d\mu$ on the space of 
continuous  real-valued functions $C(\Omega)$, 
where $\Omega$ is a finite (simply connected compact)
volume in $\R^d$, $d\geq 1\,$.

\subsubsection{Covariance of a Gaussian measure} 

We recall here the definition of the covariance of a Gaussian measure, details can be found in \cite{Glimm1987}.

A Gaussian measure of mean zero is uniquely characterized by its 
covariance $C(x,y)$
    \begin{equation}
        \int d\mu_C(\phi)\,\phi(x)\phi(y)=\tilde C(x,y)= \tilde C(y,x)\;.
    \end{equation}
$\tilde C\,$ is a positive non-degenerate bilinear form defined 
on $\mathcal{C}^\infty(\Omega)\times\mathcal{C}^\infty(\Omega)\,$. 
We assume that $\tilde C(x,y)$ is translation invariant, then 
$C(z):= \tilde C(x,y)\,, \ z=x-y\,$, is well defined. 
Using the notations
\begin{equation}
    \langle\phi,J\rangle=\int_\Omega d^dx\,\phi(x)J(x)\;,
\quad \langle J,CJ\rangle=\int_\Omega d^dx d^dy\, J(x)C(x-y)J(y)
\end{equation}
with $J\in\mathcal{C}^\infty (\Omega)$, 
the generating functional of the correlation functions is
\begin{equation}
    \int d\mu_C(\phi)e^{\langle\phi,J\rangle}=e^{\frac{1}{2}\langle J,CJ\rangle}\;.
\end{equation}
The generating functional is also called the characteristic functional 
of the Gaussian measure $\mu_C$.
For $C=(-\Delta+I)^{-1}$, where $\Delta$ denotes the Laplacian operator 
in $\R^d$, the corresponding Gaussian measure $\mu_C$ is supported on 
distributions with $1-\frac{d}{2}-\varepsilon$ continuous derivatives, 
$\varepsilon>0$. For a regularized propagator, the Fourier transform
of which falls off rapidly in momentum space, 
the Gaussian measure is supported on smooth functions. 

\subsubsection{Properties of Gaussian measures}
 We list here some properties of Gaussian measures. Proofs can be found in 
\cite{Glimm1987}.
 \begin{itemize}
     \item Integration by parts: Let $A(\phi)$ be a polynomial in 
$\phi(x)$ and its derivatives $\partial_\mu\phi(x)$. 
     \begin{equation}\label{IPP}
        \int d\mu_C(\phi) \phi(x)A(\phi)
=\int d\mu_C(\phi)\int_\Omega dy\;C(x-y)\dfrac{\delta}{\delta\phi(y)}A(\phi)\;.
    \end{equation}
    \item Translation of a Gaussian measure: Let $C$ be a covariance. 
Under a change of variable $\phi=\varphi¨+\psi$ for 
$\varphi\in\mbox{supp}(\mu_C)$ and $\psi$ such that its Fourier transform
$\hat \psi(p)$ is compactly supported.
    \begin{equation}\label{Translation_GM}
    d\mu_C(\phi)
=e^{-\frac{1}{2}\langle\psi,C^{-1}\psi\rangle}
e^{-\langle C^{-1}\psi,\varphi\rangle}d\mu_C(\varphi)\;.
\end{equation}
\item Decomposition of the covariance: Assume that
\[C=C_1+C_2\;,\quad C_i>0\;.\]
Then for $A(\phi)$ as in (\ref{IPP})
\begin{equation}
    \int d\mu_C(\phi)A(\phi)
=\int d\mu_{C_1}(\phi_1)\int d\mu_{C_2}(\phi_2)A(\phi_1+\phi_2)\;.
\end{equation}
\item Infinitesimal change of covariance: We assume the covariance 
depends on a parameter $t$, and is differentiable w.r.t. $t$
\[
C(x-y)\equiv C_t(x-y)\;,\quad \Dot{C}_t(x-y):=\dfrac{d}{dt}C_t(x-y)
\;.\]
Let $F(\phi)$ be a smooth functional, integrable w.r.t. 
$\mu_{C_t}$ $\forall t\,$. We have
\begin{equation}\label{differentiation}
    \dfrac{d}{dt}\int d\mu_{C_t}(\phi)F(\phi)
=\dfrac{1}{2}\int d\mu_{C_t}(\phi)
\left\langle\dfrac{\delta}{\delta\phi},
\Dot{C}_t\dfrac{\delta}{\delta\phi}\right\rangle F(\phi)\;.
\end{equation}
 \end{itemize}

\subsection{Isotropic Cartesian tensors}\label{Cartesian_Tensors}

\subsubsection{Isotropic Cartesian tensors}

\begin{definition}[Cartesian tensors]
Let $X$ be an Euclidean space of finite dimension $N\geq 1$. 
We identify $X$ with its dual space $X^\ast$. 
A rank $n$-tensor \textbf{T} is an element of $\bigotimes_{i=1}^n X$. 
Assuming that  we work with an orthonormal basis, we do not need to distinguish 
the contravariant and the covariant components of a tensor. 
Then, \textbf{T} $\in \bigotimes_{i=1}^n X$ is called a Cartesian 
rank $n$ tensor. Its components are denoted by $T_{i_1 i_2\cdots i_n}$.
    
\end{definition}

\begin{definition}[Isotropic Cartesian tensors]
    A Cartesian rank $n$ tensor \textbf{T} is said to be isotropic
 if for any matrix $M\in SO(N)$
    \begin{equation}
        M_{i_1 j_1}M_{i_2 j_2}\cdots M_{i_n j_n}T_{j_1 j_2\cdots j_n}=T_{i_1 i_2\cdots i_n}\;.
    \end{equation}
\end{definition}

\begin{proposition}\label{tensor_structure}
    Let \textbf{T} be a real rank $n$-tensor, $n\in 2\N\,$. 
If \textbf{T} is symmetric and $O(N)$-invariant, then it is of the form
    \begin{equation}\label{symmetric_cartesian_tensors}
        T_{i_1 i_2\cdots i_n}=A\sum_{\sigma\in S_n}
\delta_{i_{\sigma(1)}i_{\sigma(2)}}\cdots \delta_{i_{\sigma(n-1)}i_{\sigma(n)}}\;,\quad A\in\R\;,
    \end{equation}
    where $S_n$ denotes the set of permutations in 
$\,\lbrace 1,\cdots,n\rbrace\,$.
\end{proposition}

\begin{proof}
     The most general forms of 
real isotropic Cartesian rank $n$ tensors are
    \begin{itemize}
        \item $n<N$:
        \begin{equation}\label{isotropic_cartesian_tensors_1}
            T_{i_1 i_2\cdots i_n}=\sum_{\sigma\in S_n}
\lambda_\sigma\ \delta_{i_{\sigma(1)}i_{\sigma(2)}}\cdots \delta_{i_{\sigma(n-1)}i_{\sigma(n)}}\;,
\quad \lambda_\sigma\in\R\ .
        \end{equation}
        \item $n=N$:
        \begin{equation}\label{isotropic_cartesian_tensors_2}
           T_{i_1 i_2\cdots i_n}=\sum_{\sigma\in S_n}
\lambda_\sigma\ \delta_{i_{\sigma(1)}i_{\sigma(2)}}\cdots \delta_{i_{\sigma(n-1)}i_{\sigma(n)}}
+\mu\ \varepsilon_{i_1i_2\cdots i_n}\;,
\quad \lambda_\sigma,\mu\in\R\ . 
        \end{equation}
        \item $n>N$ and $N$ even:
        \begin{equation}\label{isotropic_cartesian_tensors_3}
         \!\!\!\!   T_{i_1 i_2\cdots i_n}
=\sum_{\sigma\in S_n}\lambda_\sigma\ \delta_{i_{\sigma(1)}i_{\sigma(2)}}\cdots 
\delta_{i_{\sigma(n-1)}i_{\sigma(n)}}+\sum_{\sigma\in S_n}
\mu_\sigma\ \varepsilon_{i_{\sigma(1)}\cdots i_{\sigma(N)}}
\delta_{i_{\sigma(N+1)}i_{\sigma(N+2)}}\cdots\delta_{i_{\sigma(n-1)}i_{\sigma(n)}}\;,
        \end{equation}
        where $\lambda_\sigma,\mu_\sigma\in\R\,$.
    \end{itemize}
    Here $\varepsilon_{i_1i_2\cdots i_n}$ is the Levi-Civita tensor defined by
    \begin{equation}
        \varepsilon_{i_1 i_2\cdots i_n}=\left\{
    \begin{array}{lll}
        1 & \mbox{if}\; (i_1 i_2 \cdots i_n) 
\mbox{ is an even permutation of}\; (1,2,\cdots,n) \\
    -1 & \mbox{if}\; (i_1 i_2 \cdots i_n) 
\mbox{ is an odd permutation of}\; (1,2,\cdots,n) \\
        0 & \mbox{otherwise.}
    \end{array}
\right.
    \end{equation}
For a proof see \cite{appleby_duffy_ogden_1987}. 

If \textbf{T} is $O(N)$-invariant, it is an isotropic Cartesian tensor. 
It then takes the form (\ref{isotropic_cartesian_tensors_1}),
(\ref{isotropic_cartesian_tensors_2}),(\ref{isotropic_cartesian_tensors_3}) 
depending on $n\,$. 
We consider the reflection $R$ in the hyperplane through the origin, 
orthogonal to $\bm{e_k}$, $1\leq k\leq N$, where $\bm{e_k}$ denotes 
a canonical basis vector of $\R^N\,$. 
The matrix expression of $R$ is given by
\begin{equation}
    R_{ij}=\delta_{ij}-2\delta_{ik}\delta_{jk}\ .
\end{equation}
Then we have
\begin{equation}\label{reflection}
    R_{i_1 j_1}\cdots R_{i_N j_N}\varepsilon_{j_1\cdots j_N}
=\mbox{det}(R)\varepsilon_{i_1\cdots i_N}=-\varepsilon_{i_1\cdots i_N}\ .
\end{equation}
Then from (\ref{reflection}) and (\ref{isotropic_cartesian_tensors_1}),
(\ref{isotropic_cartesian_tensors_2}), (\ref{isotropic_cartesian_tensors_3}), 
symmetric and $O(N)$-invariant tensors take the form 
(\ref{symmetric_cartesian_tensors}).
\end{proof}

\subsubsection{Contraction of isotropic Cartesian tensors}

We recall the definition of $F_{i_1i_2\cdots i_n}$

\begin{equation}
    F_{i_1i_2\cdots i_n}:=\delta_{(i_1i_2}\delta_{i_3i_4}\cdots\delta_{i_{n-1}i_n)}
:=\frac{1}{n!}\sum_{\sigma\in S_n}\delta_{i_{\sigma(1)}i_{\sigma(2)}}
\cdots \delta_{i_{\sigma(n-1)}i_{\sigma(n)}}\;.
\end{equation}

\begin{proposition}\label{contraction}
    We have the following identities:
    \begin{equation}\label{identities_contraction}
   \sum_{j=1}^N F_{i_1i_2\cdots i_n jj}=\dfrac{N+n}{n+1}F_{i_1i_2\cdots i_n}\;,
\quad \sum_{j=1}^N \mathbb{S}\Big[F_{i_1i_2\cdots i_{n_1-1}j}F_{i_{n_1}i_{n_1+1}\cdots i_{n}j}
\Big]=F_{i_1i_2\cdots i_n}\ .
\end{equation}
\end{proposition}

\begin{proof}
    Let $F(\bm{x})$ be the generating series of $F_{i_1i_2\cdots i_n}$ defined by
    
\[ F(\bm{x}):=\sum_{n\in 2\N}\sum_{i_1,i_2,\cdots, i_n}x_{i_1}
\cdots x_{i_n}F_{i_1 i_2\cdots i_n}=\sum_{n=1}^{+\infty} \vert\bm{x}\vert^{2n}\;,
\quad \vert\bm{x}\vert^2:=\sum_{i=1}^N x_i^2,\quad \bm{x}\in\R^N \ .
\]
The result of the action of the Laplacian on $F(\bm{x})$ is
\begin{equation}\label{laplacian_tensor_1}
    \Delta F(\bm{x})=\sum_{j=1}^N\partial^2_{j}F(\bm{x})
=2N+\sum_{n\in 2\N}\sum_{i_1,i_2,\cdots,i_n}x_{i_1}\cdots x_{i_n}(n+2)(n+1)
\sum_{j=1}^N F_{i_1 i_2\cdots i_n jj}\ .
\end{equation}
On the other hand we have for $n\in \N$
\begin{equation}
    \sum_{j=1}^N\partial^2_{j}\vert\bm{x}\vert^{2n}
=2n(N+2n-2)\vert\bm{x}\vert^{2n-2}
\end{equation}
and therefore
\begin{equation}\label{laplacian_tensor_2}
\begin{split}
\Delta F(\bm{x}) &=2N+\sum_{n\in 2\N,n>2}n(N+n-2)\vert\bm{x}\vert^{n-2} \\
&=2N+\sum_{n\in 2\N}(n+2)(N+n)\vert\bm{x}\vert^n \\
&=2N+\sum_{n\in 2\N}\sum_{i_1,\cdots,i_n} x_{i_1}\cdots x_{i_n}(n+2)(N+n)
F_{i_1\cdots i_n}\ .
\end{split}
\end{equation}
Now (\ref{laplacian_tensor_1}) and (\ref{laplacian_tensor_2}) imply
\begin{equation}
     \sum_{j=1}^N F_{i_1i_2\cdots i_n jj}=\dfrac{N+n}{n+1}F_{i_1i_2\cdots i_n}\;.
\end{equation}

For the second identity in (\ref{identities_contraction}), we compute
\begin{equation}
\begin{split}
 \Vert\nabla F(x)\Vert^2:=\sum_{j=1}^N(\partial_{j}F(\bm{x}))^2 &= 4\vert 
\bm{x}\vert^2\sum_{n_1,n_2\geq 1}n_1 n_2\vert\bm{x}\vert^{2(n_1+n_2-2)}
=\sum_{n_1\in 2\N, n_2\in 2\N} n_1 n_2
\vert\bm{x}\vert^{n_1+n_2-2} \\
 &= 4\vert\bm{x}\vert^2+\sum_{n\in 2\N, n>2}
\sum_{n_1+n_2=n+2}n_1 n_2\sum_{i_1,\cdots ,i_n}x_{i_1}\cdots  x_{i_n}F_{i_1\cdots i_n}\ .   
\end{split}
\end{equation}
And on the other hand, we have
\begin{equation}
    \Vert\nabla F(x)\Vert^2 =4\vert\bm{x}\vert^2+\sum_{n\in 2\N, n>2}
\sum_{n_1+n_2=n+2}n_1 n_2\sum_{i_1,\cdots ,i_n}x_{i_1}
\cdots x_{i_n}\sum_{j=1}^N F_{(i_1\cdots i_{n_1-1}j}F_{ji_{n_1}\cdots i_n) }\;,
\end{equation}
leading to 
\begin{equation}
    \sum_{j=1}^N F_{(i_1\cdots i_{n_1-1}j}F_{ji_{n_1}\cdots i_n) } = F_{i_1\cdots i_n}\;.
\end{equation}
From the definition of a symmetric part of a tensor \textbf{T} 
in (\ref{symmetric_part_tensor}) and the fact that $\mathbb{S}$ 
is an average operator, we obtain 
\begin{equation}
    \sum_{j=1}^N\mathbb{S}\Big[F_{i_1i_2\cdots i_{n_1-1}j}F_{ji_{n_1}i_{n_1+1}
\cdots i_{n}}\Big]=F_{i_1\cdots i_n}\ .
\end{equation}
\end{proof}

\subsection{Bound on a sum}\label{appendix_C}
\begin{lemma}\label{inversesquare}
    For $n\geq 12$
\begin{equation}
    \frac{n}{n-2}\sum_{\substack{n_1+n_2=n+2\\
n_i\geq 4,n_i\in 2\N}}\dfrac{1}{n_1^2(n+2-n_1)^2}\leq\frac{1}{n^2}\;.
\end{equation}
\end{lemma}

 \begin{proof}
     First we have for $n\geq 12$
\[\sum_{\substack{n_1+n_2=n+2\\
n_i\geq 4,n_i\in 2\N}}\dfrac{1}{n_1^2(n+2-n_1)^2}\leq 
\dfrac{1}{16}\sum_{\substack{n_1+n_2=\frac{n}{2}+1\\
n_i\geq 2,n_i\in\N}}\dfrac{1}{n_1^2(\frac{n}{2}+1-n_1)^2}\;.\]
We use the decomposition
\[\dfrac{1}{X^2(X-A)^2}=\dfrac{1}{A^2}\left(\dfrac{1}{X^2}
+\dfrac{1}{(X-A)^2}+\dfrac{2}{AX}-\dfrac{2}{A(X-A)}\right),\quad A>0\;.\]
We get
\begin{equation*}
\begin{split}
&\sum_{\substack{n_1+n_2=n+2\\
n_i\geq 4,n_i\in 2\N}}\dfrac{1}{n_1^2(n+2-n_1)^2}\\
&\leq\dfrac{1}{4(n+2)^2}\sum_{2\leq n_1\leq\frac{n}{2}-1}
\Bigg(\dfrac{1}{n_1^2}+\frac{1}{(\frac{n}{2}+1-n_1)^2}
+\dfrac{2}{(\frac{n}{2}+1)n_1} 
+\dfrac{2}{(\frac{n}{2}+1)(\frac{n}{2}+1-n_1)}\Bigg)\\
&\leq \dfrac{1}{2(n+2)^2}\left(\zeta(2)-1+\frac{n-4}{n+2}\right)
\leq \frac{5}{6(n+2)^2}\ ,
\end{split}
\end{equation*}
where we used the fact that 
$\sum_{2\leq n_1\leq \frac{n}{2}-1}\frac{1}{n_1}\leq \frac{n-4}{4}\,$. 
Therefore we have for $n\geq 12$
\begin{align*}
\dfrac{n}{n-2}\sum_{\substack{n_1+n_2=n+2\\
n_i\geq 4}}\dfrac{1}{n_1^2(n+2-n_1)^2} 
&\leq\dfrac{5}{6(n+2)^2}\frac{n}{n-2} 
\leq \frac{5}{6n^2}\frac{n^2}{(n+2)^2}\frac{n}{n-2}\leq\frac{1}{n^2}\;.
\end{align*}
\end{proof}

\subsection{Derivatives of 
$\frac{f}{g}$}\label{l-th_derivative_f/g}

We prove
\begin{proposition}
    For $f,g$ smooth with $g>0$,
    \begin{equation}\label{formula_appendix_D}
        \left(\dfrac{f}{g}\right)^{(l)}=\frac{1}{g}
\left[f^{(l)}-l!\;\sum_{j=1}^{l}\frac{g^{(l+1-j)}}{(l+1-j)!}
\frac{1}{(j-1)!}\left(\dfrac{f}{g}\right)^{(j-1)}\right]\;.
    \end{equation}   
\end{proposition}

\begin{proof}
    The proof is done by induction in $l\in\N$. For $l=1$, the statement 
is easily verified. 
  Then differentiating (\ref{formula_appendix_D}) and using the 
induction hypothesis, we obtain
\begin{equation}
    \begin{split}
       \left(\dfrac{f}{g}\right)^{(l+1)} 
&= \dfrac{f^{(l+1)}}{g}-\dfrac{g' f^{(l)}}{g^2}+\frac{g'}{g^2}
\sum_{j=1}^l \binom{l}{j-1}g^{(l+1-j)}\left(\dfrac{f}{g}\right)^{(j-1)} \\
       &-\frac{1}{g}\sum_{j=1}^l \binom{l}{j-1}\Big(g^{(l+2-j)}
\left(\dfrac{f}{g}\right)^{(j-1)}+g^{(l+1-j)}\left(\dfrac{f}{g}\right)^{(j)}
\Big) \\
       &= \dfrac{f^{(l+1)}}{g}-\frac{g'}{g} 
\left(\dfrac{f}{g}\right)^{(l)}-\frac{g^{(l+1)}}{g}\dfrac{f}{g}-l 
\frac{g'}{g} \left(\dfrac{f}{g}\right)^{(l)} \\
       &-\frac{1}{g}\sum_{j=2}^l \Bigg[\binom{l}{j-1}
+\binom{l}{j-2}\Bigg]g^{(l+2-j)}\left(\dfrac{f}{g}\right)^{(j-1)} \\
       &= \frac{1}{g}\left[f^{(l+1)}-(l+1)!\;\sum_{j=1}^{l+1}
\frac{g^{(l+2-j)}}{(l+2-j)!}\frac{1}{(j-1)!}
\left(\dfrac{f}{g}\right)^{(j-1)}\right]\;,
    \end{split}
\end{equation}
where we used 
\begin{equation}
    \binom{n}{k}+\binom{n}{k-1}=\binom{n+1}{k},\quad n\in\N_0\,,\ k\in\N\;.
\end{equation}  
\end{proof}

%%%%%%%%%%%%%%%%%%%%%%%%%%%%%%%%%%%%%%%%%%%%%%%%%%%%%%%%%%%%%%%%%%%%%%%%
%%%%%%%%%%%%%%%%%%%%%%%%%%%%%%%%%%%%%%%%%%%%%%%%%%%%%%%%%%%%%%%%%%%%%%%%
%%%%%%%%%%%%%%%%%%%%%%%%%%%%%%%%%%%%%%%%%%%%%%%%%%%%%%%%%%%%%%%%%%%%%%%%
%%%%%%%%%%%%%%%%%%%%%%%%%%%%%%%%%%%%%%%%%%%%%%%%%%%%%%%%%%%%%%%%%%%%%%%%
%%%%%%%%%%%%%%%%%%%%%%%%%%%%%%%%%%%%%%%%%%%%%%%%%%%%%%%%%%%%%%%%%%%%%%%%
%%%%%%%%%%%%%%%%%%%%%%%%%%%%%%%%%%%%%%%%%%%%%%%%%%%%%%%%%%%%%%%%%%%%%%%%
%%%%%%%%%%%%%%%%%%%%%%%%%%%%%%%%%%%%%%%%%%%%%%%%%%%%%%%%%%%%%%%%%%%%%%%%
%%%%%%%%%%%%%%%%%%%%%%%%%%%%%%%%%%%%%%%%%%%%%%%%%%%%%%%%%%%%%%%%%%%%%%%%
%%%%%%%%%%%%%%%%%%%%%%%%%%%%%%%%%%%%%%%%%%%%%%%%%%%%%%%%%%%%%%%%%%%%%%%%
%%%%%%%%%%%%%%%%%%%%%%%%%%%%%%%%%%%%%%%%%%%%%%%%%%%%%%%%%%%%%%%%%%%%%%%%
%%%%%%%%%%%%%%%%%%%%%%%%%%%%%%%%%%%%%%%%%%%%%%%%%%%%%%%%%%%%%%%%%%%%%%%%
%%%%%%%%%%%%%%%%%%%%%%%%%%%%%%%%%%%%%%%%%%%%%%%%%%%%%%%%%%%%%%%%%%%%%%%%
%%%%%%%%%%%%%%%%%%%%%%%%%%%%%%%%%%%%%%%%%%%%%%%%%%%%%%%%%%%%%%%%%%%%%%%%

\section{Appendix B}\label{Appendix_D}

In this appendix, we prove the different lemmas
 and propositions stated in Sect.\ref{Massive_theory}.
The bounds we obtain are expressed in terms of positive constants
$C\,,\ C_i\,$, $i=1,\ldots,11\,$ 
chosen sufficiently large and then  for $K$ sufficiently
large, depending on these constants.

\subsection{Bounds on the functions $H(\mu)\,,\ h(\mu)$}
\label{Appendix_D_Flow_equations}

Here we prove bounds on the functions $H(\mu)\,,\ h(\mu)$ 
introduced in (\ref{H(mu)})-(\ref{Hh}), and on their derivatives.

\begin{lemma}\label{boundinverseofH}
\begin{equation}
0 <     \dfrac{1}{H(\mu)}\leq C\ , \quad \mu\in[0,\Tilde{{\mu}}_{\max}]\ .
\label{bdH}
\end{equation}
\end{lemma}

\begin{proof}
We recall that we can choose $\alpha_0\leq\frac{1}{2m^2}$ so that 
$\beta_0\leq\frac{1}{2}$. Obviously $H(\mu)>0$ 
for $\mu\in[0,\tilde{{\mu}}_{\max}]$. For $\mu\in [0,-\ln(2\beta_0)]$, 
we have $H(\mu)\geq c(1+\beta_0-\frac{1}{2})\geq \frac{c}{2}>0\,$. 
For $\, \mu\in [-\ln(2\beta_0),{\mu}_{\max}]\,$ we have
\begin{equation}\label{Lemma 2.2}
    H(\mu)\geq h(\mu)\geq\dfrac{1}{4}
\int_k\dfrac{e^{-\frac{3 k^2}{2}}}{k^2+1}
\geq \dfrac{1}{4e^{\frac{3}{2}}}\int_{k,\vert k\vert\leq 1} 
\dfrac{1}{k^2+1}\geq \dfrac{1}{8e^{\frac{3}{2}}}\int_{k,\vert k\vert\leq 1} 
1 =\dfrac{c}{16e^{\frac{3}{2}}}\ .
\end{equation}
Choosing $\,C:=256\,e^{\frac{3}{2}}\,\pi^2\,$, the bound (\ref{bdH}) is satisfied.    
\end{proof}

\begin{lemma}\label{derivatives_h}
    For $l\geq 0$ and $\mu\in [0,\tilde{{\mu}}_{\max}]$
\begin{equation}\label{bound_h_k}
    \vert\partial_\mu^l h(\mu)\vert\leq c\,(5\,e)^l \vert l-1\vert!\ .
\end{equation}
\end{lemma}
    
\begin{proof}
 The proof is by induction in $l\geq 0\,$. 
First note $\vert h(\mu)\vert\leq c\,$ for $\,\mu\in[0,\tilde{{\mu}}_{\max}]\,$. 
For $l=1$
\begin{equation*}
    \vert \partial_\mu h(\mu)\vert=\vert(2+\beta_0e^\mu)h(\mu)
-c\beta_0e^\mu\vert\leq (4+2\beta_0)c\leq c\,5\,e\  ,
\end{equation*}
since $\beta_0\leq\frac{1}{2}\,$. Using Leibniz' s rule, 
we obtain for $l\geq 1$
\begin{equation}\label{Induction_Lemma2.1}
    \partial_\mu^l h(\mu)=\sum_{0\leq l_1\leq l-1}
\binom{l-1}{l_1}(2\delta_{l_1,0}+\beta_0e^\mu)
\partial_\mu^{l-1-l_1}h(\mu)-c\,\beta_0\,e^\mu\ .
\end{equation}
Inserting the induction hypothesis in the r.h.s of 
(\ref{Induction_Lemma2.1}) we get
\begin{equation}\label{inequality1}
\begin{split}
    \vert\partial_\mu^l h(\mu)\vert &\leq \sum_{0\leq l_1\leq l-1}
\binom{l-1}{l_1}(3+\beta_0)(5e)^{l-1-l_1}\vert l-2-l_1\vert!\;c+c(1+\beta_0) \\
    &\leq (3+\beta_0)(5e)^{l-1}\sum_{0\leq l_1\leq l-1}
\binom{l-1}{l_1}(l-1-l_1)!\;c+c\,(1+\beta_0) \\
    &\leq (3+\beta_0)(5e)^{l-1}(l-1)!\;e\,c+c(1+\beta_0) \\
    &\leq c\,(5e)^l (l-1)!
\left(\frac{7}{10}+\frac{3}{2\cdot 8^{l}}\right)\leq c\, (5e)^l (l-1)!\ .
\end{split}
\end{equation}
\end{proof}

\begin{lemma}\label{derivatives_H}
For $l\geq 0$, $\mu\in[0,\tilde{{\mu}}_{\max}]\,$,
\begin{equation}
    \vert\partial_\mu^l H(\mu)\vert\leq 3\,c\,(5e)^l \vert l-1\vert!\ .
\end{equation}
\end{lemma}

\begin{proof}
   From Lemma \ref{derivatives_h}
   \begin{equation}
%\begin{split}
  \vert\partial_\mu^l H(\mu)\vert \,\leq\, \vert c\,(1+\beta_0)\,\delta_{l,0}
\,-\,c\,\beta_0\ e^\mu\vert+\vert\partial_\mu^l h(\mu)\vert 
\,\leq\, 3\,c+\,c\,(5e)^l 
\vert l-1\vert!
\,\leq\, 3\,c\,(5e)^l \vert l-1\vert!\ .
%\end{split}
\end{equation}
\end{proof}

\begin{lemma}\label{derivatives_logH}
     For $l\geq 1$ and $\mu\in[0,\tilde{{\mu}}_{\max}]\,$,
    \begin{equation}
     \vert\partial_\mu^l\log(H(\mu))\vert\,\leq\, 
c\, C^l\, (5e)^{l+1}\,2^l\, (l-1)!\ . 
    \end{equation}
\end{lemma}

\begin{proof}
  %  The proof is done by induction in $l\geq 1$. 
For $l=1$ the bounds derived for $h,h'$ 
in Lemmas \ref{boundinverseofH} and \ref{derivatives_H} give :
\begin{equation}
    \vert\partial_\mu\log H(\mu)\vert\,\leq\, c\,C\,
15\, e\,\leq\, c\, C\, 50\ e^2\ .
\end{equation}
For $l\geq 1$  we 
have using (\ref{derivative_of_f/g_formula})
\begin{equation}\label{log_induction}
    \partial_\mu^{l+1} \log H(\mu) = \partial_\mu^l 
\left(\frac{H'(\mu)}{H(\mu)}\right) 
=  \dfrac{1}{H(\mu)}\left[\partial_\mu^{l+1} H(\mu)-l!\sum_{j=1}^{l}
\dfrac{\partial_\mu^{l+1-j}H(\mu)}{(l+1-j)!\; (j-1)!}
\partial_\mu^{j-1}\left(\dfrac{H'(\mu)}{H(\mu)}\right)\right]\ .
 \end{equation}
Since $\partial_\mu^{j-1}\left(\dfrac{H'(\mu)}{H(\mu)}\right)=\partial_\mu^j 
\log H(\mu)$, we can  proceed inductively on the r.h.s of  
(\ref{log_induction}). Using Lemma \ref{derivatives_H}, we get
\begin{equation}
    \begin{split}
        \vert\partial_\mu^{l+1} \log H(\mu)\vert 
&\leq C\Bigl[c\,(5e)^{l+1}l!+l!
\sum_{j=1}^l\dfrac{4c (5e)^{l+1-j}(l-j)!}{(l+1-j)!\; (j-1)!}cC^{j}(5e)^{j+1}2^j(j-1)!
\Bigr]\\
        &\leq c\,C^{l+1}(5e)^{l+2}
\Bigl[\dfrac{l!}{5eC^l}+4cl!\;\sum_{j=1}^{l}2^j\Bigr]\\
        &\leq c\,C^{l+1}(5e)^{l+2}2^{l+1}l! 
\Bigl[\frac{1}{10eC}+\frac{1}{4\pi^2}\Bigr]
\leq c\ C^{l+1}\ (5e)^{l+2}\ 2^{l+1}l!\ .
    \end{split}
\end{equation}
\end{proof}

%%%%%%%%%%%%%%%%%%%%%%%%%%%%%%%%%%%%%%%%%%%%%%%%%%%%%%%%%%%%%%%%%%%%%%%%
%%%%%%%%%%%%%%%%%%%%%%%%%%%%%%%%%%%%%%%%%%%%%%%%%%%%%%%%%%%%%%%%%%%%%%%%
%%%%%%%%%%%%%%%%%%%%%%%%%%%%%%%%%%%%%%%%%%%%%%%%%%%%%%%%%%%%%%%%%%%%%%%%
%%%%%%%%%%%%%%%%%%%%%%%%%%%%%%%%%%%%%%%%%%%%%%%%%%%%%%%%%%%%%%%%%%%%%%%%
%%%%%%%%%%%%%%%%%%%%%%%%%%%%%%%%%%%%%%%%%%%%%%%%%%%%%%%%%%%%%%%%%%%%%%%%
%%%%%%%%%%%%%%%%%%%%%%%%%%%%%%%%%%%%%%%%%%%%%%%%%%%%%%%%%%%%%%%%%%%%%%%%
%%%%%%%%%%%%%%%%%%%%%%%%%%%%%%%%%%%%%%%%%%%%%%%%%%%%%%%%%%%%%%%%%%%%%%%%
%%%%%%%%%%%%%%%%%%%%%%%%%%%%%%%%%%%%%%%%%%%%%%%%%%%%%%%%%%%%%%%%%%%%%%%%
%%%%%%%%%%%%%%%%%%%%%%%%%%%%%%%%%%%%%%%%%%%%%%%%%%%%%%%%%%%%%%%%%%%%%%%%
%%%%%%%%%%%%%%%%%%%%%%%%%%%%%%%%%%%%%%%%%%%%%%%%%%%%%%%%%%%%%%%%%%%%%%%%
%%%%%%%%%%%%%%%%%%%%%%%%%%%%%%%%%%%%%%%%%%%%%%%%%%%%%%%%%%%%%%%%%%%%%%%%
%%%%%%%%%%%%%%%%%%%%%%%%%%%%%%%%%%%%%%%%%%%%%%%%%%%%%%%%%%%%%%%%%%%%%%%%
%%%%%%%%%%%%%%%%%%%%%%%%%%%%%%%%%%%%%%%%%%%%%%%%%%%%%%%%%%%%%%%%%%%%%%%%
%%%%%%%%%%%%%%%%%%%%%%%%%%%%%%%%%%%%%%%%%%%%%%%%%%%%%%%%%%%%%%%%%%%%%%%%
%%%%%%%%%%%%%%%%%%%%%%%%%%%%%%%%%%%%%%%%%%%%%%%%%%%%%%%%%%%%%%%%%%%%%%%%
%%%%%%%%%%%%%%%%%%%%%%%%%%%%%%%%%%%%%%%%%%%%%%%%%%%%%%%%%%%%%%%%%%%%%%%%
%%%%%%%%%%%%%%%%%%%%%%%%%%%%%%%%%%%%%%%%%%%%%%%%%%%%%%%%%%%%%%%%%%%%%%%%
%%%%%%%%%%%%%%%%%%%%%%%%%%%%%%%%%%%%%%%%%%%%%%%%%%%%%%%%%%%%%%%%%%%%%%%%
%%%%%%%%%%%%%%%%%%%%%%%%%%%%%%%%%%%%%%%%%%%%%%%%%%%%%%%%%%%%%%%%%%%%%%%%
%%%%%%%%%%%%%%%%%%%%%%%%%%%%%%%%%%%%%%%%%%%%%%%%%%%%%%%%%%%%%%%%%%%%%%%%

\subsection{Bounds on the coefficients $\,\tilde{f}_{n,k}\,\ \tilde{g}_{n,k}\,$}
\label{Appendix_D_lemmas_triviality}

\factorization*
\begin{proof}
    The proof is done by induction in $N=n+2l$, going up in $l$. 
We start at $N=6$ and we have from the boundary conditions 
(\ref{BDYconditions_massive})
    \[
\tilde{f}_6(0)=0\ .
\]
    For $0\leq l<\frac{n}{2}-3$, we use (\ref{FE-mu_new}) 
and we solve it for $\partial_\mu^{l+1}\tilde{f}_n(0)$. 
Using the induction hypothesis, we obtain 
$\partial_\mu^{l+1}\tilde{f}_n(0)=0$ since in the products
    \[
\partial_\mu^{l_1} \tilde{f}_{n_1}(0)\partial_\mu^{l_2} \tilde{f}_{n_2}(0)\;,
\]
    the constraints $n_1+n_2=n+2$ and $l_1+l_2\leq l$ and $l<\frac{n}{2}-3$ 
imply that either $l_1\leq \frac{n_1}{2}-3\,$ or 
$\,l_2\leq \frac{n_2}{2}-3\,$.
\end{proof}

\boundslemma*
\begin{proof}
    From (\ref{FE_tilde_f_2_k}) we have
\begin{equation*}
     \vert \tilde{f}_{2,1}\vert
= \vert 3 \tilde{g}_{4,0}+(1-\beta_0)
\tilde{f}_{2,0}-2\Tilde{f}_{2,0}^2\vert\leq \dfrac{K}{2}\;,
\end{equation*}
and from (\ref{true_tilde_g_n_1})
\begin{equation*}
    \vert \tilde{g}_{4,1}\vert=2\vert \tilde{g}_{4,0}\vert
 \left\vert 4\tilde{f}_{2,0}+\beta_0-\frac{1}{2}h_0\right\vert
\leq \dfrac{K}{32}
\end{equation*}
choosing $\sqrt{K}>7\,c\,C>4\,$ since we have the sharper bound
\[h_0\leq 7\,c\,C\;,\]
which can be easily obtained from the explicit expression of $H$.

We proceed by induction in $n$. 
For $n\geq 6$, we find from (\ref{tilde_g_n,0}) and from 
(\ref{true_tilde_g_n_1})
for $K$ large enough
\[
\vert \tilde{g}_{n,0}\vert\leq \dfrac{n}{n-4}\frac{1}{2}\sum_{\substack{n_1+n_2=n+2\\
n_i\geq 4}}\dfrac{K^{\frac{n}{2}-\frac{3}{2}+1-\frac{3}{2}}}{n_1^2(n+2-n_1)^2}
\leq\dfrac{K^{\frac{n}{2}-\frac{3}{2}}}{2n^2}\ ,
\]
\[
\vert \tilde{g}_{n,1}\vert \leq \dfrac{2n}{n-2}(n-2)
\sum_{\substack{n_1+n_2=n+2\\
n_i\geq 4}}\dfrac{K^{\frac{n}{2}-\frac{3}{2}+2-\frac{3}{2}}}{n_1^2(n+2-n_1)^2}
+\dfrac{n}{n-2}\dfrac{K^{\frac{n}{2}-\frac{3}{2}}}{2n^2}
\left(\dfrac{\sqrt{K}}{4}+2+\frac{1}{2}7cC\right) 
\leq\dfrac{K^{\frac{n}{2}-\frac{1}{2}}}{n}\ .
\]
For $n\leq 10$ the previous bounds for the sum over $n_1$ can be checked 
explicitly, for $n\geq 12$ we use Lemma \ref{inversesquare} 
in Appendix \ref{appendix_C}.
\end{proof}

\trivprop*
\begin{proof}
    The proof proceeds by  induction in $N=n+k$ 
going up in $k$, as in the proof of Proposition \ref{triviality_prop}. 
For $k\leq 1$, the bounds follow from Lemma \ref{tilde_lemma_bounds}. 
In the r.h.s of (\ref{FE_tilde_g_n_k}), 
the first, second, fifth and seventh terms can be treated 
as in the proof of Proposition \ref{triviality_prop}. 
So we focus on the remaining terms.
    \begin{itemize}
        \item Third term: we separate the terms summed over 
$2\leq\nu\leq k-2$ and the remaining terms. 
Using Lemma \ref{tilde_lemma_bounds} we have, choosing $K>\,10\,e\,C$

        \begin{itemize}
        \item $\nu=0$:
        \begin{equation}
        \begin{split}
           \dfrac{n-2}{n+2k}\vert \tilde{g}_{n,0}h_{k+1}\vert
&\leq \dfrac{n-2}{n+2k} \dfrac{K^{\frac{n}{2}-\frac{3}{2}}}{2n^2} 
c\,C^{k+2} (5e)^{k+3} 2^{k+2} \\
&\leq 
K^{\frac{n}{2}+k+\frac{1}{2}}\left(\frac{n}{4}+k-1\right)!
\;\dfrac{1}{[(k+2)!]^{\frac{1}{8}}}5\,e\,c\ .   
        \end{split}
        \end{equation}

        \item $\nu=1$:
        \begin{equation}
            \begin{split}
             \dfrac{n-2}{n+2k}\vert \tilde{g}_{n,1}h_{k}\vert &\leq 
\dfrac{n-2}{n+2k}\dfrac{K^{\frac{n}{2}-\frac{1}{2}}}{n} c\,C^{k+1}(5e)^{k+2}2^{k+1} \\
&\leq K^{\frac{n}{2}+k+\frac{1}{2}}
\left(\frac{n}{4}+k-1\right)!\;\dfrac{1}{[(k+2)!]^{\frac{1}{8}}}5\,e\,c\ .   
            \end{split}
        \end{equation}

        \end{itemize}

        For $k-1\leq \nu\leq k+1$, we have the following bounds
        \begin{equation}
         \vert \tilde{g}_{n,\nu}\vert\leq 
K^{\frac{n}{2}+\nu-\frac{3}{2}}\left(\frac{n}{4}+k-1\right)!\;\frac{1}{[(k+2)!]^{\frac{1}{8}}}\;.   
        \end{equation} 
Using (\ref{expansion_log(H)}) we obtain
                \[
\dfrac{n-2}{n+2k}\sum_{\nu=k-1}^{k+1}\vert \Tilde{g}_{n,\nu}h_{k+1-\nu}\vert
\leq 15ec 
K^{\frac{n}{2}+k+\frac{1}{2}}\left(\dfrac{n}{4}+k-1\right)!
\;\dfrac{1}{[(k+2)!]^{\frac{1}{8}}}\ .
\]
         
        The remaining sum can be bounded for $k\geq 4$ using
 Lemma \ref{useful_lemma} and (\ref{expansion_log(H)})
        
        \begin{equation}\label{tilde_3_term}
        \begin{split}
         \dfrac{n-2}{n+2k}\sum_{\nu=2}^{k-2}\vert 
\Tilde{g}_{n,\nu} h_{k+1-\nu}\vert 
&\leq \dfrac{n-2}{n+2k}K^{\frac{n}{4}+k+\frac{1}{2}}
5ec\sum_{\nu=2}^{k-2}
\frac{(\frac{n}{4}+\nu-3)!}{[\nu!]^\frac{1}{8}} \\
         &\leq \dfrac{n-2}{n+2k}
K^{\frac{n}{4}+k+\frac{1}{2}}10ec\; 
F\Bigl(n,4,k-2,2,2,\frac{1}{8}\Bigr) \\
         &\leq \dfrac{n-2}{n+2k}
K^{\frac{n}{4}+k+\frac{1}{2}}10ec
\dfrac{2^{\frac{7}{8}}[(k-2)!]^\frac{7}{8}}{n}
\dfrac{\left(\frac{n}{4}+k-3\right)!}{(k-4)!} \\
         &\leq K^{\frac{n}{4}+k+\frac{1}{2}}\,20\,e\,c
\left(\dfrac{n}{4}+k-1\right)!\;\dfrac{1}{[(k+2)!]^{\frac{1}{8}}}\ .
        \end{split}
        \end{equation}

        Then the third term is bounded by
        \begin{equation}\label{Tilde_final_bound_third_term}
            K^{\frac{n}{2}+k+\frac{1}{2}}
\dfrac{\left(\frac{n}{4}+k-1\right)!\;}{[(k+2)!]^{\frac{1}{8}}}\frac{2}{5}\ .
        \end{equation}
        \item Fourth term: we proceed similarly as for the third term, 
then the fourth term is bounded by
        \begin{equation}\label{Tilde_final_bound_fourth_term}
            K^{\frac{n}{2}+k+\frac{1}{2}}
\dfrac{\left(\frac{n}{4}+k-1\right)!\;}{[(k+2)!]^{\frac{1}{8}}}
\frac{{C}_2}{\sqrt{K}} \ .
        \end{equation}
               \item Sixth term: Looking at the terms corresponding to 
$\nu\geq k+1$ and using the bounds from Lemma \ref{tilde_lemma_bounds} we get
        \begin{itemize}
            \item $\nu=k+1$
           
\[                         \dfrac{2}{(k+1)!}\sum_{\substack{n_1+n_2=n+2\\
n_i\geq 4}}\vert \Tilde{g}_{n_1,0}\Tilde{g}_{n_2,1}\vert 
\leq\dfrac{1}{(k+1)!}\sum_{\substack{n_1+n_2=n+2\\
n_i\geq 4}} \dfrac{K^{\frac{n}{2}-\frac{3}{2}+2-\frac{3}{2}}}{n_1^2(n+2-n_1)}
\]
\[
\leq \dfrac{1}{\sqrt{K}}
\dfrac{K^{\frac{n}{2}+k+\frac{1}{2}}}{(k+1)!}\dfrac{1}{n}\leq 
\dfrac{1}{\sqrt{K}} K^{\frac{n}{2}+k+\frac{1}{2}}
\dfrac{\left(\frac{n}{4}+k-1\right)!\;}{[(k+2)!]^{\frac{1}{8}}}\ .
\]            
            \item $\nu=k+2$
\[
             \dfrac{1}{(k+2)!}\sum_{\substack{n_1+n_2=n+2\\
n_i\geq 4}}\vert \Tilde{g}_{n_1,0}\Tilde{g}_{n_2,0}\vert 
\leq\dfrac{1}{4(k+2)!}\sum_{\substack{n_1+n_2=n+2\\
n_i\geq 4}} \dfrac{K^{\frac{n}{2}-\frac{3}{2}+1-\frac{3}{2}}}{n_1^2(n+2-n_1)^2}
\]
\[
\leq \dfrac{1}{4\sqrt{K}}\dfrac{K^{\frac{n}{2}+k+\frac{1}{2}}}{(k+2)!}
\dfrac{1}{n^2}\leq \dfrac{1}{4\sqrt{K}} K^{\frac{n}{2}+k+\frac{1}{2}}
\dfrac{\left(\frac{n}{4}+k-1\right)!\;}{[(k+2)!]^{\frac{1}{8}}}\ .
            \]
\end{itemize}
        
        For the remaining part of the sum, 
we substitute 
$F(n_1,n_2,k,c,c,\frac{1}{4})$, $c\in\lbrace 0,2\rbrace$ 
from the analysis of the third term in the proof of 
Proposition \ref{triviality_prop} by 
$F(n_1,n_2,k-\nu,c,c,\frac{1}{8})$, 
$c\in\lbrace 0,2\rbrace$. Then, the sixth term is bounded by
        \begin{equation}\label{tilde_sixth_first_sum}
            \frac{1}{2\sqrt{n+2k}}\dfrac{C_3}{\sqrt{K}}
K^{\frac{n}{2}+k+\frac{1}{2}}
\sum_{\nu=0}^{k}\frac{1}{\nu!}
\dfrac{\left(\frac{n}{4}+k-1-\nu\right)!}{[(k+2-\nu)!]^{\frac{1}{8}}}\ .
        \end{equation}
One can then extract the terms in the sum corresponding to $\nu\leq 1$ 
and bound them by 
$\frac{\left(\frac{n}{4}+k-1\right)!}{[(k+2)!]^{\frac{1}{8}}}\,$. 
The residual sum is non-zero for $k\geq 2$, it is bounded by
\[
    \sum_{\nu=2}^k\frac{1}{\nu!}
\dfrac{\left(\frac{n}{4}+k-1-\nu\right)!}{[(k+2-\nu)!]^{\frac{1}{8}}}
=\sum_{\nu=0}^{k-2}\frac{1}{(k-\nu)!}
\dfrac{\left(\frac{n}{4}-1+\nu\right)!}{[(\nu+2)!]^{\frac{1}{8}}}\leq 
\dfrac{\left(\frac{n}{4}+k-3\right)!}{[(k-2)!]^{\frac{1}{8}}}e\ .
\]
Then we obtain
\[
    \frac{1}{2\sqrt{n+2k}}\sum_{\nu=2}^k\frac{1}{\nu!}
\dfrac{\left(\frac{n}{4}+k-1-\nu\right)!}{[(k+2-\nu)!]^{\frac{1}{8}}}
\leq\frac{e}{2}\dfrac{1}{\sqrt{n+2k}}
\dfrac{\left(\frac{n}{4}+k-3\right)!}{[(k-2)!]^{\frac{1}{8}}}
\leq\frac{e}{2}\dfrac{\left(\frac{n}{4}+k-1\right)!}{[(k+2)!]^{\frac{1}{8}}}\ .
\]
Finally the sixth term is bounded by
\begin{equation}\label{Tilde_sixth_term}
    K^{\frac{n}{2}+k+\frac{1}{2}}
\dfrac{\left(\frac{n}{4}+k-1\right)!\;}
{[(k+2)!]^{\frac{1}{8}}}\frac{C_4}{\sqrt{K}}\;.
\end{equation}
\item Eighth term: first the term in the sum corresponding 
to $\nu=k+1$ is
        \[
            \dfrac{n}{n+2k}\dfrac{1}{(k+1)!}
\vert\Tilde{g}_{n,0}\tilde{f}_{2,0}\vert\leq 
\dfrac{\sqrt{K}}{16}\dfrac{K^{\frac{n}{2}-\frac{3}{2}}}{2n^2}
\leq K^{\frac{n}{2}+k+\frac{1}{2}}
\dfrac{\left(\frac{n}{4}+k-1\right)!\;}{[(k+2)!]^{\frac{1}{8}}}\frac{1}{32nK}\ .
        \]
        We follow the steps from the proof of 
Proposition \ref{triviality_prop} for the second term but substituting $F(n,4,k-2,2,2,\frac{1}{4})$ 
by $F(n,4,k-2-\nu,2,2,\frac{1}{8})$. Thus we have
        \[
     \dfrac{n}{n+2k}F\left(n,4,k-2-\nu,2,2,\frac{1}{8}\right) \leq
     \dfrac{n}{n+2k} [(k-2)!]^{\frac{7}{8}}\dfrac{4}{n}
\dfrac{\left(\frac{n}{4}+k-\nu-4\right)!}{(k-\nu-4)!}
\]
\[ 
\leq 4 \;\dfrac{\left(\frac{n}{4}+k-\nu-1\right)!}
{[(k+2-\nu)!]^{\frac{1}{8}}}\dfrac{\Big(k(k+1)(k+2)\Big)^{\frac{1}{8}}}{n+2k} 
     \leq \frac{4}{\sqrt{n+2k}} \;
\dfrac{\left(\frac{n}{4}+k-1-\nu\right)!}{[(k+2-\nu)!]^{\frac{1}{8}}}\ .
\]
Therefore we can bound the eighth term by
         \begin{equation}
          \dfrac{K^{\frac{n}{2}+k+\frac{1}{2}}\,{C}_5}{\sqrt{K}\sqrt{n+2k}}
\sum_{\nu=0}^k\dfrac{1}{\nu!}
\dfrac{\left(\frac{n}{4}+k-1-\nu\right)!}{[(k+2-\nu)!]^{\frac{1}{8}}}\ .  
         \end{equation}
       From the analysis of the sixth term, 
we deduce that the eighth term is finally bounded by
        \begin{equation}\label{Tilde_eixth_term}
    K^{\frac{n}{2}+k+\frac{1}{2}}
\dfrac{\left(\frac{n}{4}+k-1\right)!\;}{[(k+2)!]^{\frac{1}{8}}}
\frac{C_6}{\sqrt{K}}\;.
\end{equation}
    \end{itemize}
    Summing (\ref{Tilde_final_bound_third_term}), 
(\ref{Tilde_final_bound_fourth_term}), (\ref{Tilde_sixth_term}) 
and (\ref{Tilde_eixth_term}) with the already treated terms in 
the proof of Proposition \ref{triviality_prop}, we obtain 
\begin{equation}
    \vert \tilde{g}_{n,k+2}\vert 
\leq \left[\frac{2}{5}+\frac{C_7}{K}
+\frac{C_8}{\sqrt{K}}\right]
\dfrac{K^{\frac{n}{2}+k+\frac{1}{2}}}{[(k+2)!]^{\frac{1}{8}}}
\left(\dfrac{n}{4}+k-1\right)!
\leq \dfrac{K^{\frac{n}{2}+k+\frac{1}{2}}}{[(k+2)!]^{\frac{1}{8}}}
\left(\dfrac{n}{4}+k-1\right)!\ .    
\end{equation}

  The bounds for $\Tilde{f}_{2,k}$, for $k\leq 1$, follow 
from Lemma \ref{tilde_lemma_bounds}. The bounds for $k\leq 6$ can be checked by hand noting that we can always factor out 
$\dfrac{1}{\sqrt{K}}$ in the r.h.s of (\ref{FE_tilde_f_2_k}) 
using  the induction hypothesis so that the bound holds 
choosing $K$ sufficiently large. 
For $k> 6$ we will focus on the last two terms in the r.h.s of 
(\ref{FE_tilde_f_2_k}) since the other terms are treated as
in the proof of Proposition \ref{triviality_prop}.
\begin{itemize}
    \item Fourth term: it is bounded by
    \begin{equation}
        \dfrac{\beta_0}{(k+1)}
\dfrac{K^{k+\frac{1}{2}+1}}{K}
\sum_{\nu=0}^k\dfrac{\vert\nu-3\vert!}{(k-\nu)!\;
[\vert\nu-1\vert!]^{\frac{1}{8}}}\ .
    \end{equation}
    The sum can be bounded as follows
    \begin{equation}
    \begin{split}
     \sum_{\nu=0}^k\dfrac{\vert\nu-3\vert!}{(k-\nu)!\;
[\vert\nu-1\vert!]^{\frac{1}{8}}} 
&=\sum_{\nu=3}^k\dfrac{(\nu-3)!}{(k-\nu)!\;[(\nu-1)!]^{\frac{1}{8}}}+\dfrac{6}{k!}
+\dfrac{2}{(k-1)!}+\dfrac{1}{(k-2)!} \\
     &\leq [(k-1)!]^{\frac{7}{8}}\sum_{\nu=3}^k
\dfrac{(\nu-3)!}{(k-\nu)!\;(\nu-1)!}\\
     &\leq [(k-1)!]^{\frac{7}{8}} 
\dfrac{15}{k}\leq \dfrac{15(k-2)!}{[(k-1)!]^{\frac{1}{8}}}\ .
    \end{split}
    \end{equation}
    Taking into account the factor $\frac{1}{k+1}$, 
the fourth term is bounded by
    \begin{equation}\label{tildef_2_fourth}
        K^{k+1+\frac{1}{2}}\dfrac{(k-2)!}{[k!]^{\frac{1}{8}}}\dfrac{C_9}{K}\ .
    \end{equation}
    \item Fifth term: we separate the sum as follows:
    \[
\dfrac{1}{2(k+1)}\sum_{\nu=0}^{k-2}\frac{1}{\nu!}
\sum_{\nu'=0}^{k-\nu}\vert\tilde{f}_{2,\nu'}
\Tilde{f}_{2,k-\nu-\nu'}\vert
+\dfrac{1}{(k+1)(k-1)!}\vert\Tilde{f}_{2,0}\Tilde{f}_{2,1}\vert
+\dfrac{1}{2(k+1)!}\Tilde{f}_{2,0}^2\ .
\]
The bounds for the last two terms follow from Lemma \ref{tilde_lemma_bounds}. 
The sum is bounded as in the proof of Proposition \ref{triviality_prop} by a term proportional to
\begin{equation}
    \dfrac{ K^{k+1+\frac{1}{2}}}{\sqrt{K}}
\sum_{\nu=0}^{k-2}\dfrac{1}{\nu!}\dfrac{(k-2-\nu)!}{[(k-\nu)!]^{\frac{1}{8}}}\;.
\end{equation}
This sum is then treated as in the fourth term. 
Therefore we have the final bound
    \begin{equation}\label{tildef_2_fifth}
K^{k+1+\frac{1}{2}}\dfrac{(k-2)!}{[k!]^{\frac{1}{8}}}\, \dfrac{C_{10}}{\sqrt{K}}\ .
    \end{equation}
\end{itemize}
Then, we obtain 
\begin{equation}
    \vert \tilde{f}_{2,k+1}\vert\leq 
K^{k+1+\frac{1}{2}}\dfrac{(k-2)!}{[k!]^{\frac{1}{8}}}
\dfrac{C_{11}}{\sqrt{K}}\leq K^{k+1+\frac{1}{2}}
\dfrac{(k-2)!}{[k!]^{\frac{1}{8}}}\ .
\end{equation}
\end{proof}

\end{appendices}

\section*{Acknowledgement}
\paragraph{} The authors are much indebted to Riccardo Guida
for his remarks and suggestions.

\printbibliography
\end{document}